\documentclass[10pt,twocolumn,twoside]{IEEEtran}

\usepackage{amsmath}
\usepackage{amssymb}
\usepackage{amsthm}
\usepackage{graphicx}
\usepackage{epstopdf}
\usepackage{bm}
\usepackage{cite}
\usepackage{subfigure}
\newtheorem{theorem}{Theorem}[section]
\newtheorem{lemma}[theorem]{Lemma}

\newtheorem{definition}{Definition}[section]

\newtheorem{remark}[theorem]{Remark}
\newtheorem{example}[theorem]{Example}

\allowdisplaybreaks

\begin{document}
\title{Remote State Estimation  over Packet Dropping Links in the Presence of an Eavesdropper} 

\author{Alex S. Leong, Daniel E. Quevedo, Daniel Dolz, and Subhrakanti Dey
 \thanks{A. Leong and D. Quevedo are with the Department of Electrical Engineering  (EIM-E), Paderborn University, Paderborn, Germany.  E-mail: {\tt alex.leong@upb.de, dquevedo@ieee.org.} 
 D. Dolz is with Procter \& Gamble, Germany. E-mail: {\tt ddolz@uji.es.}
   S. Dey is with the Department of Engineering Science, Uppsala University, Uppsala, Sweden. E-mail: {\tt Subhra.Dey@signal.uu.se.} }
 }

\maketitle

\begin{abstract}
This paper studies remote state estimation in the presence of  an eavesdropper. A sensor transmits local state estimates over a packet dropping link to a remote estimator, while an eavesdropper can successfully overhear each sensor transmission with a certain probability. The objective is to determine when the sensor should transmit, in order to minimize the estimation error covariance at the remote estimator, while trying to keep the eavesdropper error covariance  above a certain level. This is done by solving an optimization problem that minimizes a linear combination of the expected estimation error covariance and the negative of the expected eavesdropper error covariance. Structural results on the optimal transmission policy are derived, and shown to exhibit thresholding behaviour in the estimation error covariances. In the infinite horizon situation, it is shown that with unstable systems one can keep the expected estimation error covariance bounded while the expected eavesdropper error covariance becomes unbounded. An alternative measure of security, constraining the amount of information revealed to the eavesdropper, is also considered, and similar structural results on the optimal transmission policy are derived. In the infinite horizon situation with unstable systems, it is now shown that for any transmission policy which keeps the expected estimation error covariance bounded, the expected amount of information revealed to the eavesdropper is always lower bounded away from zero. An extension of our results  to the transmission of measurements is also presented.
\end{abstract}

\section{Introduction}
With  the ever increasing amounts of data being transmitted wirelessly, the need to protect systems from malicious agents has become increasingly important. Traditionally, information security has been studied in the context of cryptography. However, due to the often limited computational power available at the transmitters (e.g. sensors in wireless sensor networks) to implement strong encryption, as well as the increased computational power available to malicious agents, achieving security using solely cryptographic methods may not be sufficient. Thus, alternative ways to implement security using information theoretic and physical layer techniques, complementary to the traditional cryptographic approaches, have attracted significant recent interest \cite{special_issue_PROC15}. 

In communications theory, the notion of information theoretic security has been around for many years, in fact dating back to the work of Claude Shannon in the 1940s \cite{Shannon_secrecy}. Roughly speaking, a communication system is regarded as secure in the information theoretic sense if the mutual information between the original message and what is received at the eavesdropper is either zero or becomes vanishingly small as the block length of the codewords increases \cite{Wyner_wiretap}. The term ``physical layer security'' has been used  to describe ways to implement information theoretic security using physical layer characteristics of the wireless channel such as fading, interference, and noise, see e.g. \cite{BlochBarros,LiangPoorShamai,GopalaLaiElGamal,ZhouSongZhang}. 

Motivated in part by the ideas of physical layer security, the consideration of security issues in signal processing systems has also started to gain the attention of researchers. For a survey on works in detection and estimation in the presence of eavesdroppers, focusing particularly on detection, see \cite{KailkhuraNadendlaVarshney}. In estimation problems with eavesdroppers, studies include \cite{AysalBarner,ReboredoXavierRodrigues,GuoLeongDey_TAES,GuoLeongDey_SIPN}. The objective is to minimize the average mean squared error at the legitimate receiver, while trying to keep the  mean squared error at the eavesdropper above a certain level, by using techniques such as stochastic bit flipping \cite{AysalBarner}, transmit filter design \cite{ReboredoXavierRodrigues}, and power control \cite{GuoLeongDey_TAES,GuoLeongDey_SIPN}. 
The above works deal with estimation of either constants or i.i.d. sources. In contrast, the focus  of the current paper is to consider the more general, and more difficult, problem of state estimation of \emph{dynamical systems} when there is an eavesdropper. For unstable systems, it has recently been shown that when using uncertain wiretap channels, one can keep the estimation error of the legitimate receiver bounded while the estimation error of the eavesdropper becomes unbounded for sufficiently large coding block length \cite{WieseJohansson}. In the current work we are interested primarily in estimation performance, and as such we do not assume coding, which can introduce large delays. Nonetheless, as we shall show, similar behaviour to \cite{WieseJohansson} can also be derived for our setup in the infinite horizon case. In a similar setup to the current work, but transmitting measurements and without using feedback acknowledgements, \cite{TsiamisGatsisPappas} derived mechanisms for keeping the expected error covariance bounded while driving the expected eavesdropper covariance unbounded, provided the reception probability is greater than the eavesdropping probability. By allowing for feedback, in this work we show that the same behaviour can be achieved for \emph{all} eavesdropping probabilities strictly less than one.

In information security, the two main types of attacks are generally regarded as: 1) passive attacks from eavesdroppers, and 2) active attacks such as Byzantine attacks or Denial of Service attacks. This paper is concerned with passive attacks from eavesdroppers. However, estimation and control problems in the presence of active attacks have also been  studied. Works in this area include \cite{LiuNingReiter,FawziTabuadaDiggavi,TeixeiraShamesSandbergJohansson,MoSinopoli_secure_estimation,BaiPasqualettiGupta,LiShi_jamming,LiQuevedo_TCNS}, just to mention a few. Another related area deals with privacy issues in estimation and control, see \cite{LeNyPappas,Cortes_differential_privacy} and the references therein.

\begin{figure}[t!]
\centering 
\includegraphics[scale=0.42]{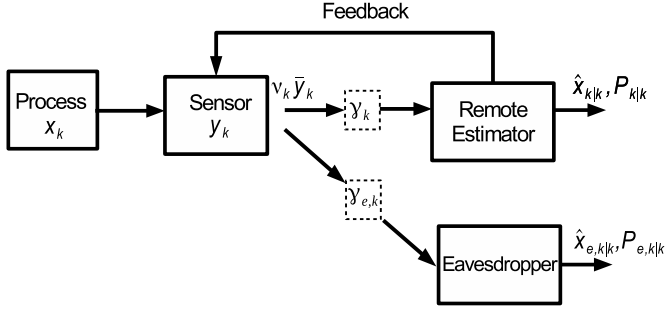} 
\caption{Remote State Estimation with an Eavesdropper}
\label{system_model}
\end{figure} 
In this paper, a sensor makes noisy measurements of a linear dynamical process. The sensor transmits  local state estimates  to the remote estimator over a packet dropping link. At the same time, an eavesdropper can successfully eavesdrop on the sensor transmission with a certain probability, see Fig. \ref{system_model}. Within this setup, we consider the problem of dynamic transmission scheduling, i.e. deciding at each instant whether the sensor should transmit. We seek to minimize a linear combination of the expected error covariance at the remote estimator and the negative of the expected error covariance at the eavesdropper. This scheduling is done at the remote estimator.

\emph{Summary of contributions}: The main contributions of this paper are:
\begin{itemize}
\item Structural results on the optimal transmission policy are derived. In the case where knowledge of the eavesdropper's error covariances are available at the remote estimator, our results show  that 1) for a fixed value of the eavesdropper's error covariance, the optimal policy has a threshold structure: the sensor should transmit if and only if the remote estimator's error covariance exceeds a certain threshold, and 2) for a fixed value of the remote estimator's error covariance, the sensor should not transmit if and only if the eavesdropper's error covariance is above a certain threshold. Such threshold policies are similar to schemes considered in event triggered estimation, e.g. \cite{LiLemmonWang,TrimpeDAndrea_journal,XiaGuptaAntsaklis,WuJiaJohanssonShi}. In the case where knowledge of the eavesdropper's error covariances are unavailable at the remote estimator, for a fixed belief of the eavesdropper's error covariance, the sensor should transmit if and only if the remote estimator's error covariance exceeds a certain threshold. 
\item For unstable systems, it is shown that in the infinite horizon situation there exist transmission policies which can keep the expected estimation error covariance bounded while the expected eavesdropper error covariance is unbounded. This behaviour can be achieved for all eavesdropping probabilities strictly less than one. 
\item An alternative measure of security, constraining the amount of information revealed to the eavesdropper (measured via the sum of conditional mutual informations), is also considered, and similar structural results on the optimal transmission policy are derived. 
\item For this alternative measure of security, in the infinite horizon situation with unstable systems, it is now shown that for any transmission policy which keeps the expected estimation error covariance bounded, the expected amount of information revealed to the eavesdropper is always lower bounded away from zero. 
\item An extension to the transmission of measurements is  described, where it is shown that threshold-type behaviour in the optimal transmission policy holds for scalar systems, but not in general for vector systems. 
\end{itemize}

This paper is organized as follows. Section \ref{model_sec} describes the system model. Section \ref{full_CSI_sec} considers the case where knowledge of the eavesdropper's error covariances is available at the remote estimator,  while Section \ref{partial_CSI_sec} studies the case where this information is unavailable. Section \ref{alt_measures_sec} considers an alternative measure of security which tries to minimize the estimator error covariance while constraining the amount of information revealed to the eavesdropper.  Section \ref{tx_meas_sec} considers the transmission of measurements. Numerical studies are given in Section \ref{numerical_sec}. Section \ref{conclusion_sec} draws conclusions. 

\section{System Model}
\label{model_sec}
A diagram of the system model is shown in Fig. \ref{system_model}.
Consider a discrete time process 
\begin{equation}
\label{state_eqn}
x_{k+1}=A x_k + w_k
\end{equation}
where $x_k \in \mathbb{R}^{n_x}$ and $w_k$ is i.i.d. Gaussian with zero mean and covariance $Q > 0$.\footnote{For a symmetric matrix $X$, we say that $X> 0 $ if it is positive definite, and  $X \geq 0 $ if it is positive semi-definite.}
The sensor has measurements
\begin{equation}
\label{measurement_eqn}
y_{k} = C x_k + v_{k}, 
\end{equation}
where $y_{k} \in \mathbb{R}^{n_y}$ and $v_{k}$ is Gaussian with zero mean and covariance $R > 0$. 
The noise processes $\{w_k\}$ and $\{v_{k}\}$ are assumed to be mutually independent.

The sensor transmits quantities $\bar{y}_k$ to the remote estimator. Common choices for $\bar{y}_k$ are the measurements, i.e. $\bar{y}_k = y_k$, or the local state estimates $\bar{y}_k = \hat{x}_{k|k}^s$ \cite{XuHespanha}. We first treat the case where the local state estimates are transmitted, see Section \ref{tx_meas_sec} for the case where measurements are transmitted. 
This requires the sensor to have some computational capabilities (i.e. the sensor is ``smart'') to run a local Kalman filter. 
The local state estimates and error covariances
\begin{equation*}
\begin{split}
\hat{x}_{k|k-1}^s & \triangleq \mathbb{E}[x_k|y_{0},\dots,y_{k-1}],  \quad
\hat{x}_{k|k}^s  \triangleq \mathbb{E}[x_k|y_{0},\dots,y_{k}] \\
P_{k|k-1}^s & \triangleq  \mathbb{E}[(x_k-\hat{x}_{k|k-1}^s)(x_k-\hat{x}_{k|k-1}^s)^T|y_{0},\dots,y_{k-1}]\\
P_{k|k}^s & \triangleq  \mathbb{E}[(x_k-\hat{x}_{k|k}^s)(x_k-\hat{x}_{k|k}^s)^T|y_{0},\dots,y_{k}]
\end{split}
\end{equation*}
can be computed at the sensor using the standard  Kalman filtering equations, see e.g. \cite{AndersonMoore}. We will assume that the pair $(A,C)$ is detectable and the pair $(A,Q^{1/2})$ is stabilizable.
Let $\bar{P}^+$ be the steady state value of $P_{k|k-1}^s$, and $\bar{P}$ be the steady state value of $P_{k|k}^s$ as $k \rightarrow \infty$, which both exist due to the detectability assumption. 
To simplify the presentation, we will assume that this local Kalman filter is operating in the steady state regime, so that   $P_{k|k}^s = \bar{P}, \forall k$. In general, the local Kalman filter will converge to steady state at an exponential rate \cite{AndersonMoore}.

Let $\nu_{k} \in \{0,1\}$ be  decision variables such that $\nu_{k}=1$ if and only if $\hat{x}_{k|k}^s$ is to be transmitted at time $k$. The decision variables $\nu_{k}$ are determined at the remote estimator, which is assumed to have more computational capabilities than the sensor, using information available at time $k-1$, and then fed back to the sensor before transmission at time $k$. 
 
At time instances when $\nu_{k}=1$, the sensor  transmits its local state estimate $\hat{x}_{k|k}^s$ over a packet dropping channel to the remote estimator.
Let $\gamma_{k}$ be random variables such that $\gamma_{k}=1$ if the sensor transmission at time $k$ is successfully received by the remote estimator, and $\gamma_{k}=0$ otherwise. We will assume that $\{\gamma_{k}\}$ is i.i.d. Bernoulli  \cite{Sinopoli} with 
$$\mathbb{P}(\gamma_{k}=1) = \lambda \in (0,1).$$

The sensor transmissions can be overheard by an eavesdropper over another packet dropping channel. Let $\gamma_{e,k}$ be random variables such that $\gamma_{e,k}=1$ if the sensor transmission at time $k$ is overheard by the eavesdropper, and $\gamma_{e,k}=0$ otherwise. We will assume that $\{\gamma_{e,k}\}$ is i.i.d. Bernoulli with 
$$\mathbb{P}(\gamma_{e,k}=1) = \lambda_e \in (0,1).$$ 
The processes $\{\gamma_{k}\}$ and $\{\gamma_{e,k}\}$ are assumed to be mutually independent.

At instances where $\nu_{k}=1$, it is assumed that the remote estimator knows whether the transmission was successful or not, i.e., the remote estimator knows the value $\gamma_{k}$, with dropped packets discarded. Define 
\begin{equation*}
\begin{split}
\mathcal{I}_k \triangleq & \{\nu_{0},\dots,\nu_{k}, \nu_{0} \gamma_{0},\dots,\nu_{k} \gamma_{k}, \nu_{0} \gamma_{0} \hat{x}_{0|0}^s,\dots,\nu_{k} \gamma_{k} \hat{x}_{k|k}^s\}
\end{split}
\end{equation*}
as the information set available to the remote estimator at time $k$. 
 Denote the state estimates and error covariances at the remote estimator by:
 \begin{equation}
\label{remote_estimator_no_tx} 
\begin{split}
\hat{x}_{k|k-1} & \triangleq \mathbb{E}[x_k| \mathcal{I}_{k-1}], \,\, \hat{x}_{k|k}  \triangleq \mathbb{E}[x_k| \mathcal{I}_k], \\ 
P_{k|k-1} & \triangleq  \mathbb{E}[(x_k-\hat{x}_{k|k-1})(x_k-\hat{x}_{k|k-1})^T|\mathcal{I}_{k-1}], \\
 P_{k|k} & \triangleq  \mathbb{E}[(x_k-\hat{x}_{k|k})(x_k-\hat{x}_{k|k})^T|\mathcal{I}_k] .
\end{split}
\end{equation}
Similarly, the eavesdropper knows if it has eavesdropped sucessfully.  Define 
\begin{equation*}
\begin{split}
\mathcal{I}_{e,k} \triangleq & \{\nu_{0},\dots,\nu_{k}, \nu_{0} \gamma_{e,0},\dots,\nu_{k} \gamma_{e,k},\\
& \quad \nu_{0} \gamma_{e,0} \hat{x}_{0|0}^s,\dots,\nu_{k} \gamma_{e,k} \hat{x}_{k|k}^s\}
\end{split}
\end{equation*}
as the information set available to the eavesdropper at time $k$, and the state estimates and error covariances at the eavesdropper by\footnote{We will assume that the eavesdropper knows the system parameters $A,C,Q,R$. If these are unknown, then the performance at the eavesdropper will be worse than the derived results.}:
 \begin{equation*}
\begin{split}
\hat{x}_{e,k|k-1} & \triangleq \mathbb{E}[x_k| \mathcal{I}_{e,k-1}], \,\, \hat{x}_{e,k|k}  \triangleq \mathbb{E}[x_k| \mathcal{I}_{e,k}], \\ 
P_{e,k|k-1} & \triangleq  \mathbb{E}[(x_k-\hat{x}_{e,k|k-1})(x_k-\hat{x}_{k|e,k-1})^T|\mathcal{I}_{e,k-1}], \\ P_{e,k|k} & \triangleq  \mathbb{E}[(x_k-\hat{x}_{e,k|k})(x_k-\hat{x}_{e,k|k})^T|\mathcal{I}_{e,k}] .
\end{split}
\end{equation*}
For simplicity of presentation, we will assume that the initial covariances $P_{0|0} = \bar{P}$ and $P_{e,0|0} = \bar{P}$.  

As stated before, the decision variables $\nu_{k}$ are determined at the remote estimator and fed back to the sensor. In Section \ref{full_CSI_sec} we consider the case where  $\nu_{k}$  depends on both $P_{k-1|k-1}$ and $P_{e,k-1|k-1}$, while in Section \ref{partial_CSI_sec} we consider the case where  $\nu_{k}$  depends only on  $P_{k-1|k-1}$ and the remote estimator's belief of $P_{e,k-1|k-1}$ constructed from knowledge of previous $\nu_{k}$'s. 
In either case, the decisions do not depend on the state $x_k$ (or the noisy measurement $y_k$). Thus, the optimal remote estimator can be shown to have the form
\begin{equation}
\label{remote_estimator_eqns_multi_sensor}
\begin{split}
\hat{x}_{k|k} & = \left\{\begin{array}{ccc}  A \hat{x}_{k-1|k-1} & , & \nu_{k} \gamma_{k} = 0 \\ \hat{x}_{k|k}^s & , & \nu_{k} \gamma_{k} = 1 \end{array}  \right. \\
P_{k|k} & = \left\{\begin{array}{ccl}  f(P_{k-1|k-1}) & ,  & \nu_{k} \gamma_{k} = 0\\ \bar{P} & , &  \nu_{k} \gamma_{k} = 1  \end{array} \right. 
\end{split}
\end{equation} 
where 
\begin{equation}
\label{f_defn}
f(X) \triangleq A X A^T + Q,
\end{equation}
while  at the eavesdropper the optimal estimator has the form
\begin{equation*}
\begin{split}
\hat{x}_{e,k|k} & = \left\{\begin{array}{ccc}  A \hat{x}_{e,k-1|k-1} & , & \nu_{k} \gamma_{e,k} = 0 \\ \hat{x}_{k|k}^s & , & \nu_{k} \gamma_{e,k} = 1 \end{array}  \right. \\
P_{e,k|k} & = \left\{\begin{array}{ccl}  f(P_{e,k-1|k-1}) & ,  & \nu_{k} \gamma_{e,k} = 0\\ \bar{P} & , &  \nu_{k} \gamma_{e,k} = 1  \end{array} \right. 
\end{split}
\end{equation*}

Define the countable set of matrices:
\begin{equation}
\label{S_defn}
\mathcal{S} \triangleq \{\bar{P}, f(\bar{P}), f^2(\bar{P}),\dots \},
\end{equation}
where $f^n(.)$ is the $n$-fold composition of $f(.)$, with the convention that $f^0(X) = X$. 
The set $\mathcal{S}$ consists of all possible values of $P_{k|k}$ at the remote estimator, as well as all possible values of  $P_{e,k|k}$ at the eavesdropper.  Given two symmetric matrices $X$ and $Y$, we say that $X\leq Y$ if $Y-X$ is positive semi-definite, and $X < Y$ if $Y-X$ is positive definite. As shown in e.g. \cite{LeongDeyQuevedo_TAC}, there is a total ordering on the elements of $\mathcal{S}$ given by 
$$ \bar{P} \leq f(\bar{P}) \leq f^2 (\bar{P}) \leq ...$$


\section{Eavesdropper Error Covariance Known at Remote Estimator}
\label{full_CSI_sec}
In this section we consider the case where the transmission decisions $\nu_k$ can depend on the error covariances of both the remote estimator  $P_{k-1|k-1}$ and the eavesdropper $P_{e,k-1|k-1}$. While knowledge of $P_{e,k-1|k-1}$ at the remote estimator may be difficult to achieve in practice, this case nevertheless serves as a useful benchmark on the achievable performance.   The  situation where  $P_{e,k-1|k-1}$  is not known at the remote estimator will be  considered in Section \ref{partial_CSI_sec}. 

We will first formulate an optimal transmission scheduling problem that minimizes a linear combination of the expected estimation error covariance and the negative of the expected eavesdropper error covariance. We then prove some structural results on the associated optimal transmission schedules.  Finally, we consider the infinite horizon situation.

\subsection{Optimal Transmission Scheduling}
The approach to security taken in Sections \ref{full_CSI_sec}-\ref{partial_CSI_sec} of this paper is to  minimize the expected error covariance at the remote estimator, while trying to keep the expected error covariance at the eavesdropper above a certain level.\footnote{Similar notions have been used in \cite{AysalBarner,ReboredoXavierRodrigues,GuoLeongDey_TAES,GuoLeongDey_SIPN}, which studied the estimation of constant parameters or i.i.d sources in the presence of an eavesdropper.} To accomplish this, we will formulate a problem that minimizes a linear combination of the expected estimation error covariance and the negative of the expected eavesdropper error covariance.  The problem we wish to solve is the finite horizon (of horizon $K$) problem:
\begin{equation}
\label{finite_horizon_problem_full_CSI_alt}
\begin{split}
& \min_{\{\nu_k\}} \sum_{k=1}^K  \mathbb{E}[\beta \textrm{tr} P_{k|k} - (1-\beta) \textrm{tr} P_{e,k|k}] 
 \\ & = \min_{\{\nu_k\}} \sum_{k=1}^K \mathbb{E} \Big[ \mathbb{E}[\beta \textrm{tr} P_{k|k} - (1-\beta) \textrm{tr} P_{e,k|k}  \\ & \quad \quad \quad \quad\quad \quad\quad \quad \quad | P_{0,0},  P_{e,0|0}, \mathcal{I}_{k-1},  \mathcal{I}_{e,k-1}, \nu_k ]\Big]
 \\ & = \min_{\{\nu_k\}} \sum_{k=1}^K \mathbb{E} \Big[ \mathbb{E}[\beta \textrm{tr} P_{k|k} - (1-\beta) \textrm{tr} P_{e,k|k}  \\ & \quad \quad \quad \quad\quad \quad\quad \quad \quad | P_{k-1,k-1},  P_{e,k-1,k-1}, \nu_k ]\Big]
 \\ & = \min_{\{\nu_k\}} \sum_{k=1}^K \mathbb{E}\Big[\beta(\nu_k \lambda \textrm{tr} \bar{P} + (1-\nu_k \lambda)\textrm{tr}f(P_{k-1|k-1}))
  \\ & \quad \quad - (1-\beta) (\nu_k \lambda_e \textrm{tr} \bar{P} + (1-\nu_k \lambda_e)\textrm{tr}f(P_{e,k-1|k-1}))\Big], 
\end{split}
\end{equation}
for some $\beta \in (0,1)$. The design parameter $\beta$ in problem (\ref{finite_horizon_problem_full_CSI_alt}) controls the tradeoff between estimation performance at the remote estimator and at the  eavesdropper, with a larger $\beta$ placing more importance on keeping $\mathbb{E}[P_{k|k}]$ small, and a smaller $\beta$ placing more importance on keeping $\mathbb{E}[P_{e,k|k}]$ large. The second equality in (\ref{finite_horizon_problem_full_CSI_alt}) holds since  $P_{k-1|k-1}$ (similarly for $P_{e,k-1|k-1}$) is a deterministic function of $P_{0|0}$ and $\mathcal{I}_{k-1}$, and $P_{k|k}$ is a function of $P_{k-1|k-1}$, $\nu_k$, and $ \gamma_{k}$. The third equality in (\ref{finite_horizon_problem_full_CSI_alt})  follows from computing the conditional expectations $\mathbb{E}[P_{k|k} | P_{k-1|k-1}, \nu_k]$ and $\mathbb{E}[P_{e,k|k} | P_{e.k-1|k-1}, \nu_k]$.

Problem (\ref{finite_horizon_problem_full_CSI_alt}) can be solved numerically using dynamic programming. 
For that purpose, define the functions $J_k(\cdot,\cdot): \mathcal{S} \times \mathcal{S} \rightarrow \mathbb{R}$ recursively as:
\begin{equation}
\label{J_fn_defn}
\begin{split}
&J_{K+1}(P,P_e)  =0 \\
&J_k(P,P_e)  = \min_{\nu \in \{0,1\}} \Big\{ \beta(\nu \lambda \textrm{tr} \bar{P} + (1-\nu \lambda)\textrm{tr}f(P)) \\ & \quad -  (1-\beta) (\nu \lambda_e \textrm{tr} \bar{P} + (1-\nu \lambda_e)\textrm{tr}f(P_e)) \\ & \quad +  \nu \lambda \lambda_e J_{k+1}(\bar{P}, \bar{P}) +  \nu \lambda (1-\lambda_e) J_{k+1}(\bar{P},f(P_e)) \\ & \quad +  \nu (1-\lambda) \lambda_e J_{k+1}(f(P),\bar{P}) 
\\ &  \quad + \big(\nu(1-\lambda)(1-\lambda_e) + 1 - \nu\big)  J_{k+1}(f(P),f(P_e))  \Big\} 
\end{split}
\end{equation}
for $k=K,\dots,1$. Then problem  (\ref{finite_horizon_problem_full_CSI_alt}) is solved by computing $J_k(P_{k-1|k-1}, P_{e,k-1|k-1})$ for $k = K,K-1,\dots,1$.  
 
\begin{remark}
Note that problem (\ref{finite_horizon_problem_full_CSI}) can be solved exactly since, for any horizon $K$, the  possible values of $(P_{k|k},P_{e,k|k})$ will lie in the finite set $  \{\bar{P}, f(\bar{P}), \dots, f^K(\bar{P}) \} \times  \{\bar{P}, f(\bar{P}), \dots, f^K(\bar{P}) \} $, which has finite  cardinality $(K+1)^2$. 
\end{remark}

\subsection{Structural Properties of Optimal Transmission Schedules}
In this subsection we will prove some structural properties on the optimal solution to problem (\ref{finite_horizon_problem_full_CSI}). In particular, we will show that 1) for a fixed $P_{e,k-1|k-1}$, the optimal policy is to only transmit if $P_{k-1|k-1}$ exceeds a threshold (which in general depends on $k$ on $P_{e,k-1|k-1}$), and 2) for a   fixed $P_{k-1|k-1}$, the optimal policy is to transmit if and only if $P_{e, k-1|k-1}$ is below a threshold (which depends on $k$ and $P_{k-1|k-1}$). Knowing that the optimal policies are of threshold-type gives insight into the form of the optimal solution, with characteristics of event triggered estimation, and can also provide computational savings when solving problem (\ref{finite_horizon_problem_full_CSI}) numerically, see \cite{Krishnamurthy_book}.

\begin{definition} 
\label{increasing_fn_defn}
A function $F(.): \mathcal{S} \rightarrow \mathbb{R}$ is \emph{increasing} if 
\begin{equation*}
X \leq Y \Rightarrow F(X) \leq F(Y).
\end{equation*}
\end{definition}

\begin{lemma}
\label{f_composition_lemma}
For any $n \in \mathbb{N}$, $\textnormal{tr} f^n(P)$ is an increasing function of $P$. 
\end{lemma}

\begin{proof}
We have 
$$\textrm{tr} f^n(P) = \textrm{tr} \left(A^n P (A^n)^T + \sum_{m=0}^{n-1} A^m Q (A^m)^T \right)$$
which is increasing with $P$. 
\end{proof}

From the definition of $J_k(.,.)$ in (\ref{J_fn_defn}), we know that if the minimizer $\nu^* =0$ then
\begin{equation}
\label{Jk_nu_0}
J_k(P,P_e)  = \beta \textrm{tr}f(P) - (1-\beta)  \textrm{tr}f(P_e) + J_{k+1}(f(P),f(P_e)),  
\end{equation} 
and if the minimizer $\nu^* =1$ then
\begin{align}
\label{Jk_nu_1}
& J_k(P,P_e) =  \beta( \lambda \textrm{tr} \bar{P} + (1- \lambda)\textrm{tr}f(P)) \nonumber \\ & \quad - (1-\beta)  ( \lambda_e \textrm{tr} \bar{P} + (1- \lambda_e)\textrm{tr}f(P_e)) \nonumber \\ & \quad +  \lambda \lambda_e J_{k+1}(\bar{P}, \bar{P})  +   \lambda (1-\lambda_e) J_{k+1}(\bar{P},f(P_e)) \nonumber  \\ & \quad +   (1-\lambda) \lambda_e J_{k+1}(f(P),\bar{P}) 
\nonumber  \\ &  \quad + (1-\lambda)(1-\lambda_e)   J_{k+1}(f(P),f(P_e)).
\end{align}
Denote  the difference of (\ref{Jk_nu_0}) and (\ref{Jk_nu_1}) as
\begin{align}
\label{phi_fn_defn}
& \phi_k(P,P_e)  \triangleq \beta \lambda  \textrm{tr}f(P) - \beta \lambda \textrm{tr} \bar{P} \nonumber \\ & \quad - (1-\beta) \lambda_e \textrm{tr}f(P_e) + (1-\beta) \lambda_e \textrm{tr} \bar{P} \nonumber \\
& \quad + [1-(1-\lambda)(1-\lambda_e)] J_{k+1}(f(P),f(P_e)) \nonumber \\ & \quad- \lambda \lambda_e J_{k+1}(\bar{P}, \bar{P}) -   \lambda (1-\lambda_e) J_{k+1}(\bar{P},f(P_e)) \nonumber \\ & \quad  -   (1-\lambda) \lambda_e J_{k+1}(f(P),\bar{P}) 
\end{align}
Note that when $\nu_k^*=1$, i.e. the optimal decision at time $k$ is to transmit, we have $\phi_k(P,P_e) >0$. The following result proves some structural properties of the optimal solution. Part (i) shows that for fixed $P_{e,k-1|k-1}$, the optimal policy is to transmit if and only if $P_{k-1|k-1}$ exceeds a threshold. Part (ii) shows that for fixed $P_{k-1|k-1}$, the optimal policy is to not transmit if and only if $P_{e,k-1|k-1}$ is above a threshold. 

\begin{theorem}
\label{structural_thm_full_CSI_alt}
(i) For fixed $P_{e,k-1|k-1}$, the optimal solution to problem (\ref{finite_horizon_problem_full_CSI_alt}) is a threshold policy on $P_{k-1|k-1}$ of the form
\begin{equation*}
\begin{split}
\nu_k^*(P_{k-1|k-1},P_{e,k-1|k-1}) = \left\{ \begin{array}{lcl} 0 & , & \textnormal{if } P_{k-1|k-1} \leq P_k^* \\ 1 & , & \textnormal{otherwise} \end{array} \right.
\end{split}
\end{equation*}
where the threshold $P_k^* \in \mathcal{S}$ depends on $k$ and $P_{e,k-1|k-1}$.
\\ (ii) For fixed $P_{k-1|k-1}$, the optimal solution to problem (\ref{finite_horizon_problem_full_CSI_alt}) is a threshold policy on $P_{e,k-1|k-1}$ of the form
\begin{equation*}
\begin{split}
\nu_k^*(P_{k-1|k-1},P_{e,k-1|k-1}) = \left\{ \begin{array}{lcl} 0 & , &  \textnormal{if } P_{e,k-1|k-1} \geq P_{e,k}^* \\ 1 & , & \textnormal{otherwise} \end{array} \right.
\end{split}
\end{equation*}
where the threshold $P_{e,k}^* \in \mathcal{S}$ depends on $k$ and $P_{k-1|k-1}$. 
\end{theorem}

\begin{proof}
(i) Since $\nu_k$ only takes on the two values $0$ and $1$, Theorem \ref{structural_thm_full_CSI_alt}(i) will be proved if we can show that 
the functions $\phi_k(P,P_e) $ defined in (\ref{phi_fn_defn}) are increasing functions of $P$ for $k=1,\dots,K$.
As $\textrm{tr} f(P)$ is an increasing function of $P$ by Lemma \ref{f_composition_lemma}, it is sufficient to show that 
$$[1-(1-\lambda)(1-\lambda_e)] J_{k}(f(P),f(P_e))-  (1-\lambda) \lambda_e J_{k}(f(P),\bar{P}) $$
is an increasing function of $P$ for all $k$. We will prove this using induction. In order to make the induction argument work, we will prove the slightly  more general statement that 
$$[1-(1-\lambda)(1-\lambda_e)] J_{k}(f^n(P),P_e)-  (1-\lambda) \lambda_e J_{k}(f^n(P),P_e') $$
is an increasing function of $P$ for all $k$, all $n\in \mathbb{N}$ and all $P_e, P_e' \in \mathcal{S}$. 

The case of $k=K+1$ is clear. Now assume that, for $P \geq P'$, 
\begin{equation}
\label{induction_hypothesis}
\begin{split}
& [1\!-\!(1\!-\!\lambda)(1\!-\!\lambda_e)] J_{l}(f^n(P),P_e) \!-\!  (1\!-\!\lambda) \lambda_e J_{l}(f^n(P),P_e') \\
& \!-\![1\!-\!(1\!-\!\lambda)(1\!-\!\lambda_e)] J_{l}(f^n(P'),\!P_e) \!+\!  (1\!-\!\lambda) \lambda_e J_{l}(f^n(P'),\!P_e') \\ & \geq 0
\end{split}
\end{equation}
holds for $l=K+1,K,\dots,k+1$. Then
\begin{align*}
& [1\!-\!(1\!-\!\lambda)(1\!-\!\lambda_e)] J_{k}(f^n(P),P_e) \!-\!  (1\!-\!\lambda) \lambda_e J_{k}(f^n(P),P_e') \\
& \!\!-\![1\!-\!(1\!-\!\lambda)(1\!-\!\lambda_e)] J_{k}(f^n(P'),\!P_e) \!+\!  (1\!-\!\lambda) \lambda_e J_{k}(f^n(P'),\!P_e') \\
& \geq \min_{\nu \in \{0,1\}} \Bigg\{  [1-(1-\lambda)(1-\lambda_e)] \Big\{ \beta [ \nu \lambda \textrm{tr} \bar{P}  + (1-\nu \lambda)  \\ & \quad \times
\textrm{tr} f^{n+1}(P)] - (1-\beta)  [ \nu \lambda_e \textrm{tr} \bar{P} + (1-\nu \lambda_e) \textrm{tr} f(P_e)] \\
& \quad  + \nu \lambda \lambda_e J_{k+1}(\bar{P},\bar{P}) + \nu \lambda(1-\lambda_e) J_{k+1} (\bar{P},f(P_e)) \\ & \quad 
+\nu(1-\lambda) \lambda_e J_{k+1} (f^{n+1}(P),\bar{P}) \\
& \quad  + [\nu(1-\lambda)(1-\lambda_e) + (1-\nu)]J_{k+1} (f^{n+1}(P),f(P_e)) \Big\} \\
&  -(1-\lambda) \lambda_e   \Big\{ \beta [ \nu \lambda \textrm{tr} \bar{P} + (1-\nu \lambda) \textrm{tr} f^{n+1}(P)] \\ & \quad 
- (1-\beta) [ \nu \lambda_e \textrm{tr} \bar{P} + (1-\nu \lambda_e) \textrm{tr} f(P_e')]  \\ &  \quad 
+ \nu \lambda \lambda_e J_{k+1}(\bar{P},\bar{P}) + \nu \lambda(1-\lambda_e) J_{k+1} (\bar{P},f(P_e')) \\ & \quad 
+\nu(1-\lambda) \lambda_e J_{k+1} (f^{n+1}(P),\bar{P}) \\
& \quad + [\nu(1-\lambda)(1-\lambda_e) + (1-\nu)]J_{k+1} (f^{n+1}(P),f(P_e')) \Big\} \\ 
&  -[1-(1-\lambda)(1-\lambda_e) ]  \Big\{ \beta [ \nu \lambda \textrm{tr} \bar{P} + (1-\nu \lambda) \textrm{tr} f^{n+1}(P')] \\ & \quad 
- (1-\beta) [ \nu \lambda_e \textrm{tr} \bar{P} + (1-\nu \lambda_e) \textrm{tr} f(P_e)] \\
&  \quad + \nu \lambda \lambda_e J_{k+1}(\bar{P},\bar{P}) + \nu \lambda(1-\lambda_e) J_{k+1} (\bar{P},f(P_e)) \\ & \quad 
+\nu(1-\lambda) \lambda_e J_{k+1} (f^{n+1}(P'),\bar{P}) \\
&  \quad + [\nu(1-\lambda)(1-\lambda_e) + (1-\nu)]J_{k+1} (f^{n+1}(P'),f(P_e)) \Big\} \\
&  + (1-\lambda) \lambda_e   \Big\{ \beta [ \nu \lambda \textrm{tr} \bar{P} + (1-\nu \lambda) \textrm{tr} f^{n+1}(P')] \\ & \quad 
- (1-\beta)[ \nu \lambda_e \textrm{tr} \bar{P} + (1-\nu \lambda_e) \textrm{tr} f(P_e')] \\
& \quad + \nu \lambda \lambda_e J_{k+1}(\bar{P},\bar{P}) + \nu \lambda(1-\lambda_e) J_{k+1} (\bar{P},f(P_e')) \\ & \quad 
+\nu(1-\lambda) \lambda_e J_{k+1} (f^{n+1}(P'),\bar{P}) \\
& \quad + [\nu(1-\lambda)(1-\lambda_e) + (1-\nu)]J_{k+1} (f^{n+1}(P'),f(P_e')) \Big\}  \Bigg\} \\
& = \min_{\nu \in \{0,1\}} \Bigg\{  [1-(1-\lambda)(1-\lambda_e)]  \Big\{ \beta  (1-\nu \lambda) \textrm{tr} f^{n+1}(P) \\ & \quad 
+\nu(1-\lambda) \lambda_e J_{k+1} (f^{n+1}(P),\bar{P}) \\ 
&  \quad + [\nu(1-\lambda)(1-\lambda_e) + (1-\nu)]J_{k+1} (f^{n+1}(P),f(P_e)) \Big\} \\
&  -(1-\lambda)\lambda_e  \Big\{ \beta  (1-\nu \lambda) \textrm{tr} f^{n+1}(P)  \\ & \quad 
+\nu(1-\lambda) \lambda_e J_{k+1} (f^{n+1}(P),\bar{P}) \\ 
&  \quad + [\nu(1-\lambda)(1-\lambda_e) + (1-\nu)]J_{k+1} (f^{n+1}(P),f(P_e')) \Big\} \\ 
&  - [1-(1-\lambda)(1-\lambda_e)]  \Big\{ \beta  (1-\nu \lambda) \textrm{tr} f^{n+1}(P')  \\ & \quad 
+\nu(1-\lambda) \lambda_e J_{k+1} (f^{n+1}(P'),\bar{P}) \\ 
& \quad + [\nu(1-\lambda)(1-\lambda_e) + (1-\nu)]J_{k+1} (f^{n+1}(P'),f(P_e)) \Big\}  \\ 
&  +(1-\lambda)\lambda_e  \Big\{ \beta  (1-\nu \lambda) \textrm{tr} f^{n+1}(P')  \\ & \quad 
 +\nu(1-\lambda) \lambda_e J_{k+1} (f^{n+1}(P'),\bar{P}) \\ 
& \quad + [\nu(1-\lambda)(1-\lambda_e) + (1-\nu)]J_{k+1} (f^{n+1}(P'),f(P_e')) \Big\} \\
& \geq 0
\end{align*}
where the last inequality holds (for both cases $\nu^*=0$ and $\nu^*=1$) by Lemma \ref{f_composition_lemma} and the induction hypothesis (\ref{induction_hypothesis}).

(ii) As $-\textrm{tr} f(P_e)$ is a decreasing function of $P_e$, it is now sufficient to show that 
$$[1-(1-\lambda)(1-\lambda_e)] J_{k}(f(P),f(P_e))-  \lambda (1-\lambda_e)  J_{k}(\bar{P},f(P_e)) $$
is a decreasing function of $P_e$ for all $k$: 
Using similar techniques as in the proof of part (i), we can prove by induction the slightly more general statement that 
$$[1-(1-\lambda)(1-\lambda_e)] J_{k}(P,f^n(P_e))-  \lambda (1-\lambda_e)  J_{k}(P',f^n(P_e)) $$
is a decreasing function of $P_e$ for all $k$, all $n\in \mathbb{N}$ and all $P, P' \in \mathcal{S}$. The details are omitted for brevity. 
\end{proof}

\subsection{Infinite Horizon}
\label{infinite_horizon_sec_full_CSI_alt}
We now consider the infinite horizon situation. 
Let us first give a condition on when  $\mathbb{E} [P_{k|k}]$ will be bounded. If $A$ is stable, this is always the case. In the case where $A$ is unstable, consider the policy with $\nu_k=1, \forall k$, which transmits at every time instant, and is similar to the situation where local state estimates are transmitted over packet dropping links \cite{Schenato,XuHespanha}. From the results of  \cite{XuHespanha} and \cite{Schenato} we have that $\mathbb{E} [P_{k|k}]$ is bounded if and only if  
\begin{equation}
\label{stability_condition}
\lambda > 1 - \frac{1}{|\sigma_{\max}(A)|^2},
\end{equation}
 where $|\sigma_{\max}(A)|$ is the largest magnitude of the eigenvalues of $A$ (i.e. the spectral radius of $A$). Thus condition (\ref{stability_condition}) will ensure the existence of policies which keep $\mathbb{E} [P_{k|k}]$ bounded. 

We will now show that for unstable systems, in the infinite horizon situation, there exists transmission policies which can make the expected eavesdropper error covariance unbounded while keeping the expected estimator error covariance bounded. This can be achieved for all probabilities of successful eavesdropping $\lambda_e$ strictly less than one. 

\begin{theorem}
\label{infinite_horizon_unbounded_thm}
Suppose that $A$ is unstable, and that $\lambda > 1 - \frac{1}{|\sigma_{\max}(A)|^2}$. Then for any $\lambda_e < 1$, there exist transmission policies in the infinite horizon situation such that $\limsup_{K \rightarrow \infty} \frac{1}{K} \sum_{k=1}^K \textnormal{tr}\mathbb{E}[P_{k|k}]$ is bounded and $\liminf_{K \rightarrow \infty} \frac{1}{K} \sum_{k=1}^K \textnormal{tr}\mathbb{E}[P_{e,k|k}]$ is unbounded. 
\end{theorem}

\begin{proof}
The proof is by construction of a policy with the required properties. Consider the threshold policy which transmits at time $k$ if and only if $P_{k-1|k-1} \geq f^t(\bar{P})$ for some $t \in \mathbb{N}$. Since $\lambda > 1 - \frac{1}{|\sigma_{\max}(A)|^2}$, one can show using results from Section IV-C of \cite{LeongDeyQuevedo_TAC} that $\lim_{K \rightarrow \infty} \frac{1}{K} \sum_{k=1}^K \textnormal{tr}\mathbb{E}[P_{k|k}] < \infty$ for any $t \in \mathbb{N}$. 

Now choose a horizon $K > t$. Consider the event $\omega$ where each transmission  is successfully received at the remote estimator, and unsuccessfully received by the eavesdropper. Using an argument similar to \cite{ShiEpsteinTiwariMurray}, we will show that the contribution of this event $\omega$ will already cause the expected eavesdropper covariance to become unbounded.  Under this event, and using the threshold policy  above, the number of transmissions that occur over the horizon $K$ is $\lfloor K/(t+1) \rfloor$, and the eavesdropper error covariances are given by $P_{e,k|k} = f^k(\bar{P}), k=1,\dots,K$. The probability of this event occurring is $(\lambda(1-\lambda_e))^{\lfloor K/(t+1) \rfloor}$. Let $\omega^c$ denote the complement of $\omega$. Then we have
\begin{equation*}
\begin{split}
& \frac{1}{K} \sum_{k=1}^K\textrm{tr} \mathbb{E}[ P_{e,k|k}] \\ & = \frac{1}{K} \!\sum_{k=1}^K   \textrm{tr}\mathbb{E}[ P_{e,k|k}| \omega] \! \times\! \mathbb{P}(\omega) +  \frac{1}{K} \!\sum_{k=1}^K \textrm{tr} \mathbb{E}[ P_{e,k|k}| \omega^c] \!\times \!\mathbb{P}(\omega^c) \\
 &> \frac{1}{K} \sum_{k=1}^K   \textrm{tr}\mathbb{E}[ P_{e,k|k}| \omega]  \mathbb{P}(\omega) \\
& = \frac{1}{K} \sum_{k=1}^K \textrm{tr}  \Big(A^k \bar{P} (A^k)^T + \sum_{m=0}^{k-1} A^m Q (A^m)^T \Big) \\ & \quad\quad\quad\quad \times (\lambda(1-\lambda_e))^{\lfloor K/(t+1) \rfloor} \\
& > \frac{1}{K} \textrm{tr} (A^K \bar{P} (A^K)^T)  (\lambda(1-\lambda_e))^{K/(t+1)} \\
& \rightarrow \infty \textrm{ as } K \rightarrow \infty,
\end{split}
\end{equation*}
where the last line holds if $|\sigma_{\max}(A)|  (\lambda(1-\lambda_e))^{1/2(t+1)} > 1$, or equivalently if
\begin{equation}
\label{lambda_e_condition_infinite_horizon}
\lambda_e < 1 - \frac{1}{\lambda |\sigma_{\max}(A)|^{2(t+1)} }.
\end{equation}
Since $|\sigma_{\max}(A)| > 1$,  the condition (\ref{lambda_e_condition_infinite_horizon}) will be satisfied for any $\lambda_e <1$ when $t$ is sufficiently large. As $ \frac{1}{K} \sum_{k=1}^K \textnormal{tr}\mathbb{E}[P_{k|k}] $ remains bounded for every $t \in \mathbb{N}$, the result follows. 
\end{proof}

In summary, the threshold policy which transmits at time $k$ if and only if $P_{k-1|k-1} \geq f^t(\bar{P})$, with $t$ large enough that condition  (\ref{lambda_e_condition_infinite_horizon}) is satisfied, will have the required properties.

\begin{remark}
In a similar setup but transmitting measurements and without using feedback acknowledgements, mechanisms were derived in \cite{TsiamisGatsisPappas} for making the expected eavesdropper error covariance unbounded while keeping the expected estimation error covariance bounded, under the more restrictive condition that $\lambda_e < \lambda$.  In a different context with coding over uncertain wiretap channels, it was shown in \cite{WieseJohansson} that for unstable systems one can keep the estimation error at the legitimate receiver bounded while the eavesdropper estimation error becomes unbounded for a sufficiently large coding block length. 
\end{remark}

\section{Eavesdropper Error Covariance  Unknown at Remote Estimator}
\label{partial_CSI_sec}
In order to construct $P_{e,k|k}$ at the remote estimator as per Section \ref{full_CSI_sec}, the process $\{\gamma_{e,k}\}$ for the eavesdropper's channel needs to be known, which in practice may be difficult to achieve. In this section, we consider the situation where the remote estimator knows only the probability of successful eavesdropping $\lambda_e$ and not the actual realizations $\gamma_{e,k}$. Thus the transmit decisions $\nu_k$ can only depend on $P_{k-1|k-1}$ and our beliefs of $P_{e,k-1|k-1}$ constructed from knowledge of previous $\nu_{k}$'s. We will first derive the recursion for the conditional distribution of error covariances at the remote estimator (i.e. the ``belief states'' \cite{Bertsekas_DP1}), and then consider the optimal transmission scheduling problem. 

\subsection{Conditional Distribution of Error Covariances at Eavesdropper}
\label{belief_recursion_sec}
Define the belief vector
\begin{equation}
\label{belief_vector}
\pi_{e,k} = \left[ \begin{array}{c} \pi_{e,k}^{(0)} \\ \pi_{e,k}^{(1)} \\ \vdots \\ \pi_{e,k}^{(K)} \end{array} \right] \triangleq 
\left[ \begin{array}{c} \mathbb{P}\big(P_{e,k|k} = \bar{P} | \nu_0, \dots, \nu_k \big) \\  \mathbb{P}\big(P_{e,k|k} = f(\bar{P}) | \nu_0, \dots, \nu_k \big)  \\ \vdots \\  \mathbb{P}\big(P_{e,k|k} = f^K(\bar{P}) | \nu_0, \dots, \nu_k \big)  \end{array} \right]
\end{equation}
We note  that by our assumption of $P_{e,0|0}=\bar{P}$, we have $\pi_{e,k}^{(K)} \triangleq \mathbb{P}\big(P_{e,k|k} = f^K(\bar{P}) | \nu_0, \dots, \nu_k \big) = 0$ for $k < K$. Denote the set of all possible $\pi_{e,k}$'s by $\Pi_e \subseteq \mathbb{R}^{K+1}$. 

The vector $\pi_{e,k}$ represents our beliefs on $P_{e,k|k}$ given the transmission decisions $\nu_0, \dots, \nu_k $. In order to formulate the transmission scheduling problem as a partially observed problem in the next subsection, we first want to derive a recursive relationship between $\pi_{e,k+1}$ and $\pi_{e,k}$ given the next transmission decision $\nu_{k+1}$. When  $\nu_{k+1} = 0$, we have $P_{e,k+1|k+1} = f( P_{e,k|k} )$ with probability one, and thus  $\pi_{e,k+1} = \left[\begin{array}{cccc} 0 & \pi_{e,k}^{(0)} & \dots & \pi_{e,k}^{(K-1)} \end{array} \right]^T$. When $\nu_{k+1} = 1$, then $P_{e,k+1|k+1} = \bar{P}$ with probability $\lambda_e$ and $P_{e,k+1|k+1} = f( P_{e,k|k} )$ with probability $1-\lambda_e$, and thus  
$\pi_{e,k+1} = \left[\begin{array}{cccc}  \lambda_e & (1-\lambda_e)\pi_{e,k}^{(0)}  & \dots & (1-\lambda_e) \pi_{e,k}^{(K-1)}  \end{array} \right]^T$. 

Hence, defining  
\begin{equation*}
\begin{split}
&\Phi(\pi_{e}, \nu) \\ & \triangleq  \left\{ \!\! \begin{array}{lc} 
\left[\begin{array}{cccc} 0 & \pi_{e}^{(0)} & \dots & \pi_{e}^{(K-1)} \end{array} \right]^T\!, & \nu = 0   \\ 
  \left[\begin{array}{cccc}  \lambda_e & (1\!-\!\lambda_e) \pi_{e}^{(0)}  & \dots & (1\!-\!\lambda_e) \pi_{e}^{(K-1)}  \end{array} \right]^T\!, & \nu = 1 \end{array}\right. 
\end{split}
\end{equation*}
we obtain the recursive relationship 
$$ \pi_{e,k+1 } = \Phi(\pi_{e,k}, \nu_{k+1}) .$$

\subsection{Optimal Transmission Scheduling}
We again wish to minimize a linear combination of the expected error covariance at the remote estimator and the negative of the expected error covariance at the eavesdropper. Since $P_{e,k-1|k-1}$ is not available, the optimization problem will now be formulated as a partially observed one with $\nu_k$ dependent on $(P_{k-1|k-1}, \pi_{e,k-1})$. 
We then have the following problem (cf. (\ref{finite_horizon_problem_full_CSI_alt})):
\begin{equation}
\label{finite_horizon_problem_partial_CSI_alt}
\begin{split}
&  \min_{\{\nu_k\} } \sum_{k=1}^K \mathbb{E}  \Big[\beta(\nu_k \lambda \textrm{tr} \bar{P} + (1-\nu_k \lambda)\textrm{tr}f(P_{k-1|k-1}))
 \\ & \quad - (1-\beta) \Big(\nu_k \lambda_e \textrm{tr} \bar{P} + (1-\nu_k \lambda_e) \sum_{i=0}^K \textrm{tr} f^{i+1} (\bar{P}) \pi_{e,k-1}^{(i)} \Big) \Big].
\end{split}
\end{equation}

Problem (\ref{finite_horizon_problem_partial_CSI_alt}) can be solved by using the dynamic programming algorithm for partially observed problems \cite{Bertsekas_DP1}.  Let the functions $\mathcal{J}_k(\cdot,\cdot): \mathcal{S} \times \Pi_e \rightarrow \mathbb{R}$ be defined recursively as:
\begin{equation}
\label{Jhat_fn_defn}
\begin{split}
&\mathcal{J}_{K+1}(P,\pi_e)  =0 \\
&\mathcal{J}_k(P,\pi_e)  = \min_{\nu \in \{0,1\}} \Big\{ \beta(\nu \lambda \textrm{tr} \bar{P} + (1-\nu \lambda)\textrm{tr}f(P))   \\ & \quad 
- (1-\beta) \Big(\nu \lambda_e \textrm{tr} \bar{P} + (1-\nu \lambda_e) \sum_{i=0}^K \textrm{tr} f^{i+1} (\bar{P}) \pi_{e}^{(i)} \Big) \\ & \quad +  \nu \lambda  \mathcal{J}_{k+1}\big(\bar{P}, \Phi(\pi_e,1)\big) +  \nu (1-\lambda)  \mathcal{J}_{k+1}\big(f(P),\Phi(\pi_e,1)\big) 
\\ & \quad 
+ (1 - \nu)  \mathcal{J}_{k+1}\big(f(P),\Phi(\pi_e,0)\big)  \Big\} 
\end{split}
\end{equation}
for $k=K,\dots,1$. Then problem  (\ref{finite_horizon_problem_partial_CSI_alt}) is solved numerically by computing $\mathcal{J}_k(P_{k-1|k-1}, \pi_{e,k-1})$ for $k = K,K-1,\dots,1$.

\begin{remark}
\label{belief_discret_remark}
In the finite horizon situation, the number of possible values of $(P_{k|k},\pi_{e,k})$ is again finite, but now of cardinality $(K+1) \times (1+2+\dots+2^K) = (K+1) (2^{K+1}-1)$. This is exponential in $K$, which may be very large when $K$ is large. To reduce the complexity, one could consider instead probability distributions
$$\left[ \begin{array}{c} \pi_{e,k}^{(0)} \\ \pi_{e,k}^{(1)} \\ \vdots \\ \pi_{e,k}^{(N-1)} \\ \pi_{e,k}^{(N)}\end{array} \right] \triangleq 
\left[ \begin{array}{c} \mathbb{P}\big(P_{e,k|k} = \bar{P} | \nu_0, \dots, \nu_k \big) \\  \mathbb{P}\big(P_{e,k|k} = f(\bar{P}) | \nu_0, \dots, \nu_k \big)  \\ \vdots \\  \mathbb{P}\big(P_{e,k|k} = f^{N-1}  (\bar{P})  | \nu_0, \dots, \nu_k \big)\\ \mathbb{P}\big(P_{e,k|k} \geq f^{N} (\bar{P}) | \nu_0, \dots, \nu_k \big)  \end{array} \right]$$
for some $N < K$, and  update the beliefs via:
\begin{equation*}
\begin{split}
&\Phi^N(\pi_{e}, \nu) \\ & \triangleq  \left\{ \!\!\begin{array}{lc} 
\left[\begin{array}{ccccc} 0 & \pi_{e}^{(0)} & \dots & \pi_{e}^{(N-2)}  & \pi_{e}^{(N-1)} \!\!+\! \pi_{e}^{(N)} \end{array} \right]^T\!\!, & \nu = 0   \\ 
  \left[\begin{array}{cccc}  \lambda_e & (1-\lambda_e) \pi_{e}^{(0)}  & \dots & (1-\lambda_e) \pi_{e}^{(N-2)} \end{array}\right.  \\ \left.\begin{array}{c}\quad \quad\quad \quad \quad\quad(1-\lambda_e) (\pi_{e}^{(N-1)} +    \pi_{e}^{(N)})\end{array} \right]^T\!\!, & \nu = 1 \end{array}\right. 
\end{split}
\end{equation*}
Discretizing the space of $\pi_{e,k}$ to include the cases with up to $N-1$ successive packet drops or non-transmissions, with the remaining cases grouped into the single component  $ \pi_{e,k}^{(N)}$, will then give a state space of cardinality $(K+1) (2^{N+1}-1)$. 
\end{remark}

\subsection{Structural Properties}
Denote  the difference in the values of $\mathcal{J}_k(P,\pi_e)$ when the minimizing $\nu^*$ are 0 and 1 by
\begin{equation}
\label{psi_fn_defn}
\begin{split}
&\psi_k(P,\pi_e)  \triangleq \beta \lambda  \textrm{tr}f(P) - \beta \lambda \textrm{tr} \bar{P} \\ & \quad - (1-\beta) \lambda_e \Big(\sum_{i=0}^K \textrm{tr} f^{i+1}(\bar{P}) \pi_{e}^{(i)} - \textrm{tr}  \bar{P} \Big) \\
& \quad +  \mathcal{J}_{k+1}\big(f(P),\Phi(\pi_e,0)\big)- \lambda \mathcal{J}_{k+1}\big(\bar{P},\Phi(\pi_e,1)\big) \\ & \quad  -    (1-\lambda)  \mathcal{J}_{k+1}\big(f(P), \Phi(\pi_e,1)\big). 
\end{split}
\end{equation}

\begin{theorem}
\label{fixed_Pi_e_thm_partial_CSI_alt}
For fixed $\pi_{e,k-1}$, the optimal solution  to problem (\ref{finite_horizon_problem_partial_CSI_alt}) is a threshold policy on $P_{k-1|k-1}$ of the form
\begin{equation*}
\begin{split}
\nu_k^*(P_{k-1|k-1},\pi_{e,k-1}) = \left\{ \begin{array}{lcl} 0 & , & \textnormal{if } P_{k-1|k-1} \leq P^* \\ 1 & , & \textnormal{otherwise} \end{array} \right.
\end{split}
\end{equation*}
where the threshold $P^*$ depends on $k$ and $\pi_{e,k-1}$.
\end{theorem}

\begin{proof}
 Theorem \ref{fixed_Pi_e_thm_partial_CSI_alt} will be proved by showing that 
for fixed $\pi_e$, the functions $\psi_k(P,\pi_e) $ defined by (\ref{psi_fn_defn}) are increasing functions of $P $ for $k=1,\dots,K$. This will be the case if we can show that
$$\mathcal{J}_{k}\big(f(P),\Phi(\pi_e,0)\big)-    (1-\lambda)  \mathcal{J}_{k}\big(f(P), \Phi(\pi_e,1)\big)  $$
is an increasing function of $P$ for all $k$. Using a similar induction argument as in the proof of Theorem \ref{structural_thm_full_CSI_alt}(i), we can establish the slightly more general statement that 
$$\mathcal{J}_{k}\big(f^n(P),\pi_e\big)-    (1-\lambda)  \mathcal{J}_{k}\big(f^n(P), \pi_e'\big)  $$
is an increasing function of $P$ for all $k$, all $n\in \mathbb{N}$ and all $\pi_e,\pi_e' \in \Pi_e$. 
\end{proof}

\subsection{Infinite Horizon}
In the infinite horizon situation, we note that Theorem \ref{infinite_horizon_unbounded_thm} will still hold, as the threshold policy constructed in the proof does not require knowledge of the eavesdropper error covariances.

\section{An Alternative Measure of Security}
\label{alt_measures_sec}
We have so far studied security from the viewpoint of trying to keep the eavesdropper error covariance above a certain level, which has also been used in other works such as \cite{AysalBarner,ReboredoXavierRodrigues,GuoLeongDey_TAES,GuoLeongDey_SIPN}. However other measures of security are possible (and may be more appropriate depending on the situation).
One alternative measure of security is restricting the amount of information revealed to the eavesdropper, where information is defined in an information theoretic sense \cite{CoverThomas}. More specifically, we want to restrict the sum of conditional mutual informations  (also known as  the \emph{directed  information} \cite{Massey_ISITA}), revealed to the eavesdropper. The directed information measure has also been used in control system design with data rate constraints and source coding on the feedback path \cite{SilvaDerpichOstergaard}, and joint sensor and controller design for LQG control  \cite{TanakaSandberg}. 

Let $z_{e,k} \triangleq \gamma_{e,k} \bar{y}_k = \gamma_{e,k} \hat{x}_{k|k}^s$ be the signal received by the eavesdropper. The conditional mutual information
$$I_{e,k} \triangleq I(x_k; z_{e,k} | z_{e,0},\dots,z_{e,k-1})$$
 between $x_k$ and $z_{e,k}$ has the expression (see \cite{TanakaSandberg,CoverThomas}):
\begin{equation}
\label{directed_mutual_info}
\begin{split}
I_{e,k} & = \frac{1}{2} \log \det P_{e,k|k-1} - \frac{1}{2} \log \det P_{e,k|k} \\ & = \frac{1}{2} \log \det f(P_{e,k-1|k-1}) - \frac{1}{2} \log \det P_{e,k|k} 
\end{split}
\end{equation}

\subsection{Eavesdropper Error Covariance  Known at Remote Estimator}
\label{full_CSI_sec2}
We may consider  the finite horizon problem:
\begin{equation}
\label{finite_horizon_problem_full_CSI}
\begin{split}
& \min_{\{\nu_k \} } \sum_{k=1}^K \mathbb{E}[\beta \textrm{tr} P_{k|k} + (1-\beta) I_{e,k}] 
 \\ & = \min_{\{\nu_k \} } \sum_{k=1}^K \mathbb{E}[\mathbb{E}[\beta \textrm{tr} P_{k|k} + (1-\beta) I_{e,k} \\ & \quad\quad\quad\quad\quad\quad\quad\quad|P_{k-1|k-1}, P_{e,k-1|k-1}, \nu_k]] 
 \\ & = \min_{\{\nu_k \} } \sum_{k=1}^K \mathbb{E}\Big[\beta(\nu_k \lambda \textrm{tr} \bar{P} + (1-\nu_k \lambda)\textrm{tr}f(P_{k-1|k-1})) \\ & \quad + (1-\beta) \nu_k \lambda_e \Big(\frac{1}{2} \log \det f(P_{e,k-1|k-1}) - \frac{1}{2} \log \det \bar{P}\Big)\Big], 
\end{split}
\end{equation}
for some $\beta \in (0,1)$, noting that in the computation of $\mathbb{E}[I_{e,k}| P_{k-1|k-1}, P_{e,k-1|k-1},\nu_k]$, $P_{e,k|k} = \bar{P}$ when $\nu_k=1$ and $\gamma_{e,k}=1$.  The design parameter $\beta$ in problem (\ref{finite_horizon_problem_full_CSI}) now controls the tradeoff between estimation performance at the remote estimator and the amount of information revealed to the eavesdropper, with a larger $\beta$ placing more importance on keeping $\mathbb{E}[P_{k|k}]$ small, and a smaller $\beta$ placing more importance on keeping $\mathbb{E}[I_{e,k}]$ small.

We have the following structural results:
\begin{theorem}
\label{structural_thm_full_CSI}
(i) For fixed $P_{e,k-1|k-1}$, the optimal solution to problem (\ref{finite_horizon_problem_full_CSI}) is a threshold policy on $P_{k-1|k-1}$ of the form
\begin{equation*}
\begin{split}
\nu_k^*(P_{k-1|k-1},P_{e,k-1|k-1}) = \left\{ \begin{array}{lcl} 0 & , & \textnormal{if } P_{k-1|k-1} \leq P^t_k \\ 1 & , & \textnormal{otherwise} \end{array} \right.
\end{split}
\end{equation*}
where the threshold $P^t_k \in \mathcal{S}$ depends on $k$ and $P_{e,k-1|k-1}$.
\\(ii) For fixed $P_{k-1|k-1}$, the optimal solution to problem (\ref{finite_horizon_problem_full_CSI}) is a threshold policy on $P_{e,k-1|k-1}$ of the form
\begin{equation*}
\begin{split}
\nu_k^*(P_{k-1|k-1},P_{e,k-1|k-1}) = \left\{ \begin{array}{lcl} 0 & , & \textnormal{if } P_{e,k-1|k-1} \geq P_{e,k}^t \\ 1 & , & \textnormal{otherwise} \end{array} \right.
\end{split}
\end{equation*}
where the threshold $P_{e,k}^t\in\mathcal{S}$ depends on $k$ and $P_{k-1|k-1}$. 
\end{theorem}

\begin{proof}
See Appendix \ref{structural_thm_full_CSI_proof}.
\end{proof}

In order to prove Theorem \ref{structural_thm_full_CSI}(ii), we will also need the following:
\begin{lemma}
\label{logdet_lemma}
The function
$$\log \det f^n(P) - \log \det f^{n+1}(P) $$ 
is an increasing function of $P$ for all $n \in \mathbb{N}$.
\end{lemma}

\begin{proof}
See Appendix \ref{logdet_lemma_proof}.
\end{proof}

The infinite horizon counterpart to (\ref{finite_horizon_problem_full_CSI}) is:
 \begin{equation}
\label{infinite_horizon_problem_full_CSI}
\begin{split}
& \min_{\{\nu_k \} } \limsup_{K \rightarrow \infty} \frac{1}{K} \sum_{k=1}^K \mathbb{E}[\beta \textrm{tr} P_{k|k} + (1-\beta) I_{e,k}] \\
 & = \min_{\{\nu_k \} } \limsup_{K \rightarrow \infty} \frac{1}{K} \! \sum_{k=1}^K \!\mathbb{E}\Big[\beta(\nu_k \lambda \textrm{tr} \bar{P} \!+\! (1\!-\!\nu_k \lambda)\textrm{tr}f(P_{k-1|k-1})) \\ & \quad+ (1-\beta) \nu_k \lambda_e \Big(\frac{1}{2} \log \det f(P_{e,k-1|k-1}) - \frac{1}{2} \log \det \bar{P}\Big)\Big], 
\end{split}
\end{equation}

As in Section \ref{infinite_horizon_sec_full_CSI_alt}, when $A$ is unstable $\mathbb{E} [P_{k|k}]$ can be kept  bounded if and only if  
$\lambda > 1 - \frac{1}{|\sigma_{\max}(A)|^2}$. The question now is whether one needs a similar condition on $\lambda_e$ in order to keep $\mathbb{E} [I_{e,k}]$ bounded, and hence ensure  the existence of solutions with bounded cost to problem (\ref{infinite_horizon_problem_full_CSI}). The answer turns out to be ``no'' (i.e. $ \limsup_{K \rightarrow \infty} \frac{1}{K} \sum_{k=1}^K \mathbb{E}[I_{e,k}]$ is bounded for all $\lambda_e$). To show this, for a given $K$, let the random variable $\tau(K)$ denote the number of times  where $\gamma_{e,k} = 1$ for $k \leq K$. Denote the random times between successful eavesdroppings by $N_1, N_2, \dots, N_{\tau(K)}$, with $N_1 + N_2 + \dots + N_{\tau(K)} \leq K$. Then  we have 
\begin{equation}
\label{I_e_bound}
\begin{split}
& \frac{1}{K} \sum_{k=1}^{K}  \mathbb{E}[I_{e,k}] \\& = \frac{1}{K} \mathbb{E} \left[ \sum_{i=1}^{\tau(K)}  \left(\frac{1}{2} \log \det f^{N_i}(\bar{P}) -  \frac{1}{2} \log \det \bar{P} \right) \right]\\
& =  \frac{1}{K} \mathbb{E} \left[ \sum_{i=1}^{\tau(K)} N_i \times \frac{1}{N_i}  \left(\frac{1}{2} \log \det f^{N_i}(\bar{P}) -  \frac{1}{2} \log \det \bar{P} \right) \right]\\
& <  \frac{1}{K} \mathbb{E} \left[ \sum_{i=1}^{\tau(K)} N_i  \Delta_U \right] \leq \Delta_U
\end{split}
\end{equation}
where the first inequality comes from  the following result:
\begin{lemma}
\label{infinite_horizon_upper_bound_lemma}
Let $A$ be an unstable matrix. Then there exists a $\Delta_U < \infty$, dependent only on $A, Q, \bar{P}$, such that
$$ \frac{1}{N} \left(\frac{1}{2} \log \det f^N(\bar{P}) -  \frac{1}{2} \log \det \bar{P} \right)  < \Delta_U$$
for all $N \in \mathbb{N}$. 
\end{lemma}
\begin{proof}
See Appendix \ref{infinite_horizon_upper_bound_lemma_proof}.
\end{proof}
\noindent As (\ref{I_e_bound}) holds for all $K$, we thus have $ \limsup_{K \rightarrow \infty} \frac{1}{K} \sum_{k=1}^K \mathbb{E}[I_{e,k}] < \Delta_U < \infty$.

Theorem \ref{infinite_horizon_unbounded_thm} showed that for unstable systems, one could always find policies which can keep the estimation error covariance bounded while the eavesdropper covariance became unbounded. We might ask whether using the alternative measure of security of this section, one can also keep the estimation error bounded while driving the information revealed to the eavesdropper to zero. 
Theorem \ref{infinite_horizon_info_thm}  however will show that the answer is negative: in order to keep the error covariance of the remote estimator bounded, one will always reveal a non-zero amount of information to the eavesdropper.  Before proving this fundamental result in Theorem \ref{infinite_horizon_info_thm}, we will need the following: 
\begin{lemma}
\label{infinite_horizon_lower_bound_lemma}
Let $A$ be an unstable matrix. Then there exists a  $\Delta_L > 0$, dependent only on $A, Q, \bar{P}$, such that
$$ \frac{1}{N} \left(\frac{1}{2} \log \det f^N(\bar{P}) -  \frac{1}{2} \log \det \bar{P} \right)  > \Delta_L$$
for all $N \in \mathbb{N}$. 
\end{lemma}

\begin{proof}
See Appendix \ref{infinite_horizon_lower_bound_lemma_proof}.
\end{proof}

\begin{theorem}
\label{infinite_horizon_info_thm}
Let $A$ be an unstable matrix, and assume that $\lambda_e > 0$. Then, for any transmission policy satisfying 
\begin{equation}
\label{bounded_error_covariance}
\limsup_{K \rightarrow \infty} \frac{1}{K} \sum_{k=1}^K \mathbb{E}[P_{k|k}] < \infty, 
\end{equation}
one must have 
$$\liminf_{K \rightarrow \infty}\frac{1}{K} \sum_{k=1}^K \mathbb{E}[I_{e,k}] > \epsilon$$
for some $\epsilon > 0$ dependent only on $\lambda_e, A, Q, \bar{P}$.
\end{theorem}

\begin{proof}
Arguing by contradiction, we first note that for any transmission policy satisfying (\ref{bounded_error_covariance}), the time between any two successive transmission attempts must be upper bounded by some constant $\kappa_{\max}$, where $\kappa_{\max}$ depends on the particular policy used. 

Now fix such a policy satisfying (\ref{bounded_error_covariance}). Call a sequence of transmission decisions $\boldsymbol{\nu} \triangleq \{\nu_0,\nu_1,\dots\} $  \emph{admissible} for the policy if there is a sequence  $\{\gamma_0,\gamma_1,\dots, \gamma_{e,0},\gamma_{e,1},\dots \}$ which generates it. We will show that $\frac{1}{K} \sum_{k=1}^K \mathbb{E}[I_{e,k}| \boldsymbol{\nu}] > 0$ is lower bounded away from zero for  any admissible sequence $\boldsymbol{\nu}$ and all sufficiently large $K$, thus proving the theorem. 

Fix a  $K > \kappa_{\max}$ and an admissible sequence $\boldsymbol{\nu}$. Let $K'$ denote the last time over the horizon $K$ when there is a transmission. Then $K - K' < \kappa_{\max}$, and $K'/K > 1 - \kappa_{\max} /K$ is bounded away from zero for all $K > \kappa_{\max}$.  
We have
\begin{align*}
& \frac{1}{K} \sum_{k=1}^K \mathbb{E}[I_{e,k}| \boldsymbol{\nu}] \\ & = \frac{1}{K} \sum_{k=1}^K  \mathbb{E}[I_{e,k}| \boldsymbol{\nu}, \gamma_{e,K'}=1] \times \mathbb{P}(\gamma_{e,K'}=1) \\ & \quad + \frac{1}{K} \sum_{k=1}^K\mathbb{E}[I_{e,k}| \boldsymbol{\nu}, \gamma_{e,K'}=0] \times \mathbb{P}(\gamma_{e,K'}=0) \\
& \geq \frac{\lambda_e}{K} \sum_{k=1}^K  \mathbb{E}[I_{e,k}| \boldsymbol{\nu}, \gamma_{e,K'}=1] \\
& =  \frac{\lambda_e}{K} \sum_{k=1}^{K'}  \mathbb{E}[I_{e,k}| \boldsymbol{\nu}, \gamma_{e,K'}=1] \\
& = \frac{\lambda_e K'}{K}  \frac{1}{K'} \sum_{k=1}^{K'}  \mathbb{E}[I_{e,k}| \boldsymbol{\nu}, \gamma_{e,K'}=1] \\
& > \lambda_e \left(1 - \frac{\kappa_{\max} }{K}\right) \frac{1}{K'} \sum_{k=1}^{K'}  \mathbb{E}[I_{e,k}| \boldsymbol{\nu}, \gamma_{e,K'}=1] 
\end{align*}
Now consider the term $\frac{1}{K'} \sum_{k=1}^{K'}  \mathbb{E}[I_{e,k}| \boldsymbol{\nu}, \gamma_{e,K'}=1]$. 
Within this time period $K'$, there could be a number of instances $k$ where $\gamma_{e,k} = 1$. Denote the random times between successful eavesdroppings by $N_1, N_2, \dots, N_{\tau(K')}$, where the random variable $\tau(K')$ denotes the total number of successful eavesdroppings within the time period $K'$ (since $\gamma_{e,K'}=1$, we have $\tau(K')\geq 1$), with $N_1 + N_2 + \dots + N_{\tau(K')} = K'$. Then  we have 
\begin{equation*}
\begin{split}
& \frac{1}{K'} \sum_{k=1}^{K'}  \mathbb{E}[I_{e,k}| \boldsymbol{\nu}, \gamma_{e,K'}=1] \\ & = \frac{1}{K'} \mathbb{E} \left[ \sum_{i=1}^{\tau(K')}  \left(\frac{1}{2} \log \det f^{N_i}(\bar{P}) -  \frac{1}{2} \log \det \bar{P} \right) \right]\\
& =  \frac{1}{K'} \mathbb{E} \left[ \sum_{i=1}^{\tau(K')} N_i \!\times\! \frac{1}{N_i}  \left(\frac{1}{2} \log \det f^{N_i}(\bar{P}) -  \frac{1}{2} \log \det \bar{P} \right) \right]\\
& \geq  \frac{1}{K'} \mathbb{E} \left[ \sum_{i=1}^{\tau(K')} N_i  \Delta_L \right] = \Delta_L
\end{split}
\end{equation*}
where the inequality comes from Lemma \ref{infinite_horizon_lower_bound_lemma}. 
Hence 
\begin{equation*}
\begin{split}
\frac{1}{K} \sum_{k=1}^K \mathbb{E}[I_{e,k}| \boldsymbol{\nu}] & > \lambda_e  \left(1 - \frac{\kappa_{\max} }{K} \right) \Delta_L, 
\end{split}
\end{equation*}
for all $K > \kappa_{\max}$. 
In particular,
$$\frac{1}{K} \sum_{k=1}^K \mathbb{E}[I_{e,k}| \boldsymbol{\nu}] > \lambda_e \Delta_L \triangleq \epsilon$$ for all sufficiently large $K$.
\end{proof}

We now return to the problem (\ref{infinite_horizon_problem_full_CSI}).  The Bellman equation for problem  (\ref{infinite_horizon_problem_full_CSI}) is: 
\begin{equation}
\label{Bellman_eqn}
\begin{split}
&\rho + h(P,P_e)  = \min_{\nu \in \{0,1\}} \Big\{ \beta(\nu \lambda \textrm{tr} \bar{P} + (1-\nu \lambda)\textrm{tr}f(P)) \\ & \quad + (1-\beta) \nu \lambda_e \Big(\frac{1}{2} \log \det f(P_{e}) - \frac{1}{2} \log \det \bar{P}\Big) \\ & \quad +  \nu \lambda \lambda_e h(\bar{P}, \bar{P}) +  \nu \lambda (1-\lambda_e) h(\bar{P},f(P_e)) \\ & \quad +  \nu (1-\lambda) \lambda_e h(f(P),\bar{P}) 
\\ &  \quad + \big(\nu(1-\lambda)(1-\lambda_e) + 1 - \nu\big)  h(f(P),f(P_e))  \Big\} 
\end{split}
\end{equation}
where $\rho$ is the optimal average cost per stage and $h(.,.)$ is the differential cost or relative value function \cite{Bertsekas_DP1}. Solutions $h(.,.)$ to the Bellman equation (\ref{Bellman_eqn}) can be found using the relative value iteration algorithm. 
Define the functions $V_l(\cdot,\cdot): \mathcal{S} \times \mathcal{S} \rightarrow \mathbb{R}$ as:
\begin{equation*}
\begin{split}
&V_0(P,P_e)  =0 \\
&V_{l+1}(P,P_e)  = \min_{\nu \in \{0,1\}} \Big\{ \beta(\nu \lambda \textrm{tr} \bar{P} + (1-\nu \lambda)\textrm{tr}f(P)) \\ & \quad + (1-\beta) \nu \lambda_e \Big(\frac{1}{2} \log \det f(P_{e}) - \frac{1}{2} \log \det \bar{P}\Big) \\ & \quad +  \nu \lambda \lambda_e V_l(\bar{P}, \bar{P}) +  \nu \lambda (1-\lambda_e) V_l(\bar{P},f(P_e)) \\ & \quad +  \nu (1-\lambda) \lambda_e V_l(f(P),\bar{P}) 
\\ &  \quad + \big(\nu(1-\lambda)(1-\lambda_e) + 1 - \nu\big)  V_l(f(P),f(P_e))  \Big\} 
\end{split}
\end{equation*}
for $l=0,1,2,\dots$. Fix an arbitrary $(P^f,P_e^f) \in \mathcal{S} \times \mathcal{S}$. The relative value iteration algorithm then computes 
$$h_l(P,P_e) \triangleq V_{l}(P,P_e) - V_{l}(P^f,P_e^f)$$
for $l=0,1,2,\dots$, with $h_l(P,P_e) \rightarrow h(P,P_e), \forall (P,P_e) \in \mathcal{S} \times \mathcal{S}$  as $l \rightarrow \infty$. 

\begin{remark}
\label{SN_remark}
In the infinite horizon case, the number of possible values of $(P_{k|k},P_{e,k|k})$
is infinite. Thus in practice the state space will need to be truncated for numerical solution. 
For instance, define the finite set $\mathcal{S}^N \subseteq \mathcal{S}$ by
\begin{equation*}
\mathcal{S}^N \triangleq \{\bar{P}, f(\bar{P}), \dots, f^N(\bar{P}) \},
\end{equation*}
which includes the values of all error covariances with up to $N$ successive packet drops or non-transmissions. Then one can run the relative value iteration algorithm over the finite state space $\mathcal{S}^N \times  \mathcal{S}^N$ (of cardinality $(N+1)^2)$), and compare the solutions obtained as $N$ increases to determine an appropriate value for $N$ \cite{Sennott_book}. 
\end{remark}

We have the following structural results:

\begin{theorem}
\label{infinite_horizon_thm}
(i) For fixed $P_{e,k-1|k-1}$, the optimal solution to problem (\ref{infinite_horizon_problem_full_CSI}) is a threshold policy on $P_{k-1|k-1}$ of the form
\begin{equation*}
\begin{split}
\nu_k^*(P_{k-1|k-1},P_{e,k-1|k-1}) = \left\{ \begin{array}{lcl} 0 & , & \textnormal{if } P_{k-1|k-1} \leq P^* \\ 1 & , & \textnormal{otherwise} \end{array} \right.
\end{split}
\end{equation*}
where the threshold $P^*$ depends on  $P_{e,k-1|k-1}$.\\
(ii) For fixed $P_{k-1|k-1}$, the optimal solution to problem (\ref{infinite_horizon_problem_full_CSI}) is a threshold policy on $P_{e,k-1|k-1}$ of the form
\begin{equation*}
\begin{split}
\nu_k^*(P_{k-1|k-1},P_{e,k-1|k-1}) = \left\{ \begin{array}{lcl} 0 & , & \textnormal{if } P_{e,k-1|k-1} \geq P_e^* \\ 1 & , & \textnormal{otherwise} \end{array} \right.
\end{split}
\end{equation*}
where the threshold $P_e^*$ depends on  $P_{k-1|k-1}$. 
\end{theorem}

\begin{proof}
We can  verify that the arguments used in the proof of Theorem \ref{structural_thm_full_CSI} hold when $J_{k+1}(.,.)$ is replaced by $V_l(.,.)$, and hence also holds for $h_l(.,.)$. The result then follows by noting that $h_l(P,P_e) \rightarrow h(P,P_e) \,\, \forall (P,P_e)$.
\end{proof}

\begin{remark}
In contrast to Theorem \ref{structural_thm_full_CSI}, for the infinite horizon the thresholds $P^*$ and $P_e^*$ do not depend on the time index $k$. 
\end{remark}

\begin{remark}
The motivation for considering in problems (\ref{finite_horizon_problem_full_CSI_alt}) and (\ref{finite_horizon_problem_partial_CSI_alt}) the mean squared error (or estimation error covariances) at the eavesdropper rather than the  mutual information is that since we are considering an estimation problem, using estimation theoretic measures is perhaps more suitable than information theoretic measures, which often assume large block lengths and hence large delays. Nonetheless, the two notions are closely related, e.g. the expression (\ref{directed_mutual_info}) for the conditional mutual information is given as the log determinant of the estimation error covariances. 
Other relations between information theory and estimation theory have also been discovered in the literature, see e.g. \cite{GuoShamaiVerdu,WeissmanKimPermuter,NaghibiSalimiSkoglund,FengLoparo}.
\end{remark}

\subsection{Eavesdropper Error Covariance  Unknown at Remote Estimator}
In the case where the eavesdropper error covariances are unknown to the remote estimator, we obtain the problems:
\begin{equation}
\label{finite_horizon_problem_partial_CSI}
\begin{split}
& \min_{\{\nu_k \} } \sum_{k=1}^K  \mathbb{E}\Big[\beta(\nu_k \lambda \textrm{tr} \bar{P} + (1-\nu_k \lambda)\textrm{tr}f(P_{k-1|k-1})) \\&\! +\! (1\!-\!\beta) \nu_k \lambda_e \Big(\sum_{i=0}^K \frac{1}{2} ( \log \det f^{i+1}(\bar{P})) \pi_{e,k-1}^{(i)} \!-\! \frac{1}{2} \log \det  \bar{P} \Big) \Big],
\end{split}
\end{equation}
\begin{equation}
\label{infinite_horizon_problem_partial_CSI}
\begin{split}
& \min_{\{\nu_k \} }\limsup_{K \rightarrow \infty} \frac{1}{K}  \sum_{k=1}^K \mathbb{E}\Big[\beta(\nu_k \lambda \textrm{tr} \bar{P} + (1-\nu_k \lambda)\textrm{tr}f(P_{k-1|k-1}))  \\ & \! +\! (1\!-\!\beta) \nu_k \lambda_e \Big(\sum_{i=0}^K \frac{1}{2} ( \log \det f^{i+1}(\bar{P})) \pi_{e,k-1}^{(i)} \!-\! \frac{1}{2} \log \det  \bar{P} \Big) \Big],
\end{split}
\end{equation}
for the finite horizon and infinite horizon situations respectively. Similar techniques as in Section \ref{partial_CSI_sec} can be used to analyze these problems. 

\section{Transmission of Measurements}
\label{tx_meas_sec}
In this section, we will briefly describe an extension to the transmission of measurements, which can also be analyzed using similar techniques as presented in the preceding sections. Only the finite horizon situation will be studied here.  

In the case where measurements are transmitted, the remote estimator and eavesdropper will run Kalman filters. 
At the remote estimator, we have
\begin{equation*}
\label{remote_estimator_tx_meas}
\begin{split}
\hat{x}_{k+1|k} & = A \hat{x}_{k|k}, \quad
\hat{x}_{k|k}  = \hat{x}_{k|k-1} + \gamma_{k} K_{k} (y_{k} - C \hat{x}_{k|k-1}) \\
P_{k+1|k} & = A P_{k|k} A^T + Q, \quad
P_{k|k}  = P_{k|k-1} - \gamma_{k} K_{k} C P_{k|k-1} 
\end{split}
\end{equation*} 
where $K_{k} \triangleq P_{k|k-1} C^T (C P_{k|k-1} C^T + R)^{-1} $.
We can thus write:
\begin{equation}
\label{remote_estimator_eqns_tx_meas}
\begin{split}
\hat{x}_{k\!+\!1|k} & \!= \!\left\{\begin{array}{ccc}  A \hat{x}_{k|k-1} & \!\!\!\!\!, \, \nu_{k} \gamma_{k} = 0 \\ A \hat{x}_{k|k\!-\!1} \!+\! A K_{k} (y_{k} \!-\! C \hat{x}_{k|k\!-\!1}) & \!\!\!\!\!, \, \nu_{k} \gamma_{k} = 1 \end{array}  \right. \\
P_{k\!+\!1|k} & \!=\! \left\{\begin{array}{cll} f(P_{k|k-1}) & , \,  \nu_{k} \gamma_{k} = 0 \\ g(P_{k|k-1}) & , \,  \nu_{k} \gamma_{k} = 1, \end{array} \right. 
\end{split}
\end{equation} 
where $f(X) \triangleq A X A^T + Q$ as before, and 
\begin{equation*}
g(X) \triangleq  A X A^T - A X C^T (C X C^T + R)^{-1} C X A^T + Q, 
\end{equation*}
From Kalman filtering theory, we have that $g(P)$ is an increasing function of $P$ in the sense of Definition~\ref{increasing_fn_defn}. 
Note that in (\ref{remote_estimator_eqns_tx_meas}) the recursions are given in terms of $\hat{x}_{k+1|k}$ and $P_{k+1|k}$ (rather than $\hat{x}_{k|k}$ and $P_{k|k}$), which  are slightly more convenient to work with. 
Similarly, at the eavesdropper we have 
\begin{equation*}
\label{eavesdropper_eqns_tx_meas}
\begin{split}
\hat{x}_{e,k\!+\!1|k} & \!= \!\left\{\begin{array}{ccc}  A \hat{x}_{e,k|k-1} & \!\!\!\!\!, \, \nu_{k} \gamma_{e,k} = 0 \\ A \hat{x}_{e,k|k\!-\!1} \!+\! A K_{e,k} (y_{k} \!-\! C \hat{x}_{e,k|k\!-\!1}) & \!\!\!\!\!, \, \nu_{k} \gamma_{e,k} = 1 \end{array}  \right. \\
P_{e,k\!+\!1|k} & \!=\! \left\{\begin{array}{cll} f(P_{e,k|k-1}) & , \,  \nu_{k} \gamma_{e,k} = 0 \\ g(P_{e,k|k-1}) & , \,  \nu_{k} \gamma_{e,k} = 1, \end{array} \right. 
\end{split}
\end{equation*} 
where $K_{e,k} \triangleq P_{e,k|k-1} C^T (C P_{e,k|k-1} C^T + R)^{-1} $.

\subsection{Eavesdropper Error Covariance Known at Remote Estimator}
In the case where the error covariance of the eavesdropper is available at the remote estimator, we may consider the problem (cf. (\ref{finite_horizon_problem_full_CSI_alt})):
\begin{equation}
\label{finite_horizon_problem_full_CSI_tx_meas}
\begin{split}
& \min_{\{\nu_k \} } \sum_{k=1}^K \mathbb{E}[\beta \textrm{tr} P_{k+1|k} - (1-\beta) \textrm{tr} P_{e,k+1|k}] 
 \\ 
& = \min_{\{\nu_k \} } \sum_{k=1}^K \mathbb{E}\Big[\beta\big(\nu_k \lambda \textrm{tr} g(P_{k|k-1}) + (1-\nu_k \lambda)\textrm{tr}f(P_{k|k-1})\big) \\ & \quad - (1-\beta) \big(\nu_k \lambda_e \textrm{tr} g(P_{e,k|k-1}) + (1-\nu_k \lambda_e)\textrm{tr}f(P_{e,k|k-1})\big)\Big], 
\end{split}
\end{equation}
where the  transmission decisions $\nu_k$ are dependent on $P_{k|k-1}$ and $P_{e,k|k-1}$.
For scalar systems, we have the following structural results:
\begin{theorem}
\label{structural_thm_full_CSI_tx_meas}
Suppose the system is scalar. 
\\(i) For fixed $P_{e,k|k-1}$, the optimal solution to problem (\ref{finite_horizon_problem_full_CSI_tx_meas}) is a threshold policy on $P_{k|k-1}$ of the form
\begin{equation*}
\begin{split}
\nu_k^*(P_{k|k-1},P_{e,k|k-1}) = \left\{ \begin{array}{lcl} 0 & , & \textnormal{if } P_{k|k-1} \leq P^*_k \\ 1 & , & \textnormal{otherwise} \end{array} \right.
\end{split}
\end{equation*}
where the threshold $P^*_k \in \mathbb{R}$ depends on $k$ and $P_{e,k|k-1}$.
\\(ii) For fixed $P_{k|k-1}$, the optimal solution to problem (\ref{finite_horizon_problem_full_CSI_tx_meas}) is a threshold policy on $P_{e,k|k-1}$ of the form
\begin{equation*}
\begin{split}
\nu_k^*(P_{k|k-1},P_{e,k|k-1}) = \left\{ \begin{array}{lcl} 0 & , & \textnormal{if } P_{e,k|k-1} \geq P_{e,k}^* \\ 1 & , & \textnormal{otherwise} \end{array} \right.
\end{split}
\end{equation*}
where the threshold $P_{e,k}^* \in \mathbb{R}$ depends on $k$ and $P_{k|k-1}$. 
\end{theorem}

\begin{proof}
See Appendix \ref{structural_thm_full_CSI_tx_meas_proof}
\end{proof}

\begin{example}
The following is a counterexample to show that for vector systems, Theorem \ref{structural_thm_full_CSI_tx_meas}(i) in general does not hold. Similar counterexamples to Theorem \ref{structural_thm_full_CSI_tx_meas}(ii) can also be constructed for vector systems, but for brevity will not be presented here. 

Consider a system with parameters 
$$A = \left[\begin{array}{cc} 1.2 & 0.2 \\ 0.3 & 0.8 \end{array} \right], \quad C = \left[\begin{array}{cc} 1 & -0.5 \end{array} \right], $$
$Q=I$, $R=1$, $\lambda=0.6$, $\lambda_e=0.6$. We have
$$\bar{P}^+ = \left[\begin{array}{cc} 6.2117 & 4.7680 \\ 4.7680 & 5.9176 \end{array} \right].$$

Consider the transmission decision at the final time instant $k=K$. Let 
$$P_1=\bar{P}^+, P_2 = \left[\begin{array}{cc} 6.4 & 4.5 \\ 4.5 & 6.3 \end{array} \right], P_{e} = f(\bar{P}^+).$$
The final stage costs of problem (\ref{finite_horizon_problem_full_CSI_tx_meas}) are: 
\begin{equation*}
\begin{split}
&\beta\big(\nu_K \lambda \textnormal{tr} g(P_{K|K-1}) + (1-\nu_K \lambda)\textnormal{tr}f(P_{K|K-1})\big) \\& - (1-\beta) \big(\nu_K \lambda_e \textnormal{tr} g(P_{e,K|K-1}) + (1-\nu_K \lambda_e)\textnormal{tr}f(P_{e,K|K-1})\big)
\end{split}
\end{equation*}
For $\beta=0.73$, $P_{K|K-1} = P_1$, and $P_{e,K|K-1}=P_e$, the final stage costs
when $\nu_K=0$ and $\nu_K=1$ are $5.4979$ and $5.4427$ respectively, and thus it is optimal to transmit at time $K$. 
On the other hand, for $P_{K|K-1} = P_2$ and $P_{e,K|K-1}=P_e$, the final stage costs when $\nu_K=0$ and $\nu_K=1$ are $5.7103$ and $5.9216$ respectively, and thus it is now optimal to not transmit at time $K$.  But since $P_2 > P_1$ (as one can easily verify), this shows that Theorem  \ref{structural_thm_full_CSI_tx_meas}(i) in general does not hold for vector systems. 
\end{example}

\subsection{Eavesdropper Error Covariance  Unknown at Remote Estimator}
Let us define
\begin{equation*}
\begin{split}
 \mathcal{F}_k^0 (P) &\triangleq \underbrace{g \circ g \circ \dots \circ g \circ g}_{k} (P) \\
 \mathcal{F}_k^1 (P) &\triangleq g \circ g \circ \dots \circ g \circ f (P) \\
  \mathcal{F}_k^2 (P) &\triangleq g \circ g \circ \dots \circ f \circ g (P) \\
 & \vdots \\
  \mathcal{F}_k^{2^k-1} (P) &\triangleq f \circ f \circ \dots \circ f \circ f (P) 
\end{split}
\end{equation*}
where $\circ$ denotes composition, and the ordering of the $g$'s and $f$'s in  $\mathcal{F}_k^i$ is determined by the binary representation of $i$ and the correspondence $g \rightarrow 0, f \rightarrow 1$. Assuming that the initial covariances $P_{0|0} = \bar{P}$ and $P_{e,0|0} = \bar{P}$, 
then the  possible values that $P_{k+1|k}$ can take will lie in the finite set
$\{\mathcal{F}_k^0 (\bar{P}^+), \mathcal{F}_k^1 (\bar{P}^+), \dots, \mathcal{F}_k^{2^k-1} (\bar{P}^+) \}$.

Now define the belief vector (c.f. (\ref{belief_vector}))
$$
\pi_{e,k} \! =\! \left[ \!\!\begin{array}{c} \pi_{e,k}^{(0)} \\ \pi_{e,k}^{(1)} \\ \vdots \\ \pi_{e,k}^{(2^k-1)} \end{array} \!\! \right] \! \triangleq  \!
\left[ \!\!\begin{array}{c} \mathbb{P}\big(P_{e,k+1|k} = \mathcal{F}_k^0 (\bar{P}^+) | \nu_0, \dots, \nu_k \big) \\  \mathbb{P}\big(P_{e,k+1|k} = \mathcal{F}_k^1 (\bar{P}^+) | \nu_0, \dots, \nu_k \big)  \\ \vdots \\  \mathbb{P}\big(P_{e,k+1|k} = \mathcal{F}_k^{2^k-1} (\bar{P}^+) | \nu_0, \dots, \nu_k \big)  \end{array} \!\!\right]
$$
We can easily verify that the following recursive relationship for the beliefs holds:
$$ \pi_{e,k+1 } = \Phi_k(\pi_{e,k}, \nu_{k+1}),$$
where $\Phi_k(\pi_{e}, \nu) : \mathbb{R}^{2^k} \times \{0,1\} \rightarrow  \mathbb{R}^{2^{k+1}}$ is given  by
\begin{equation*}
\begin{split}
& \Phi_k(\pi_{e}, \nu) \\ &\!  \triangleq  \left\{ \!\!\! \begin{array}{lc} 
\left[\begin{array}{cccccc} 0 & \dots & 0 & \pi_{e}^{(0)} & \dots & \pi_{e}^{(2^k-1)} \end{array} \right]^T\!\!, & \textnormal{if } \nu = 0   \\ 
  \left[\begin{array}{cccc}  \lambda_e \pi_{e}^{(0)} & \dots &  \lambda_e \pi_{e}^{(2^k-1)} & (1-\lambda_e) \pi_{e}^{(0)}  \end{array}\right. \\ \left.  \begin{array}{cc} \quad\quad\quad\quad\quad\quad\quad  \dots &(1-\lambda_e) \pi_{e}^{(2^k-1)}  \end{array} \right]^T\!\!, & \textnormal{if } \nu = 1 \end{array}\right. 
\end{split}
\end{equation*}
The transmission scheduling problem is then given by
\begin{equation}
\label{finite_horizon_problem_partial_CSI_tx_meas}
\begin{split}
& \min_{\nu_k \in \{0,1\} } \sum_{k=1}^K \mathbb{E}\Big[\beta(\nu_k \lambda \textrm{tr}g(P_{k|k-1})) + (1-\nu_k \lambda)\textrm{tr}f(P_{k|k-1})) \\ & \quad  - (1-\beta) \Big( \nu_k \lambda_e  \sum_{i=0}^{2^{k\!-\!1}\!-1}g ( \mathcal{F}_{k-1}^i(\bar{P}^+)) \pi_{e,k-1}^{(i)} \\ & \quad\quad\quad\quad + (1-\nu_k \lambda_e) \sum_{i=0}^{2^{k\!-\!1}\!-1} f( \mathcal{F}_{k-1}^i(\bar{P}^+)) \pi_{e,k-1}^{(i)}  \Big) \Big],
\end{split}
\end{equation}
where the  transmission decisions $\nu_k$ are dependent on $P_{k|k-1}$ and $\pi_{e,k-1}$.
Similar to Theorem \ref{fixed_Pi_e_thm_partial_CSI_alt}, we have the following structural result:
\begin{theorem}
\label{fixed_Pi_e_thm_partial_CSI_tx_meas}
Suppose the system is scalar. Then for fixed $\pi_{e,k-1}$, the optimal solution to problem (\ref{finite_horizon_problem_partial_CSI_tx_meas}) is a threshold policy on $P_{k|k-1}$ of the form
\begin{equation*}
\begin{split}
\nu_k^*(P_{k|k-1},\pi_{e,k-1}) = \left\{ \begin{array}{lcl} 0 & , & \textnormal{if } P_{k|k-1} \leq P^*_k \\ 1 & , & \textnormal{otherwise} \end{array} \right.
\end{split}
\end{equation*}
where the threshold $P^*_k \in \mathbb{R}$ depends on $k$ and $\pi_{e,k-1}$.
\end{theorem}

\section{Numerical Studies}
\label{numerical_sec}
We consider an example with parameters 
$$A = \left[\begin{array}{cc} 1.2 & 0.2 \\ 0.3 & 0.8  \end{array}\right],\, C = \left[\begin{array}{cc} 1 & 1 \end{array} \right],\, Q=I,\, R=1$$
The steady state error covariance $\bar{P}$ is easily computed as
 $$\bar{P} =  \left[\begin{array}{rr} 1.3411 & -0.8244 \\ -0.8244 & 1.0919  \end{array}\right].$$

\subsection{Finite Horizon}
We will here solve the finite horizon problem with $K=10$. The packet reception probability is chosen to be $\lambda=0.6$, and the eavesdropping probability $\lambda_e =0.6$.  
Assuming that  the eavesdropper error covariance is available, and using the design parameter $\beta=0.7$, Fig. \ref{threshold_plot} plots $\nu_k^*$ for different values of $P_{k-1|k-1}=f^n(\bar{P})$  and $P_{e,k-1|k-1}=f^{n_e}(\bar{P})$, at the time step $k=4$. Fig. \ref{threshold_plot2} plots $\nu_k^*$ at the time step $k=6$. We observe a threshold behaviour in both $P_{k-1|k-1}$ and $P_{e,k-1|k-1}$, with the thresholds also dependent on the time $k$, in agreement with Theorem \ref{structural_thm_full_CSI_alt}.
\begin{figure}[t!]
\centering 
\includegraphics[scale=0.5]{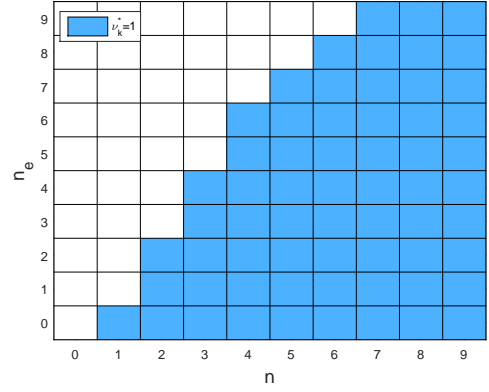} 
\caption{$\nu_k^*$ for different values of $P_{k-1|k-1}=f^n(\bar{P})$ and $P_{e,k-1|k-1}=f^{n_e}(\bar{P})$, at time $k=4$.}
\label{threshold_plot}
\end{figure}  
\begin{figure}[t!]
\centering 
\includegraphics[scale=0.5]{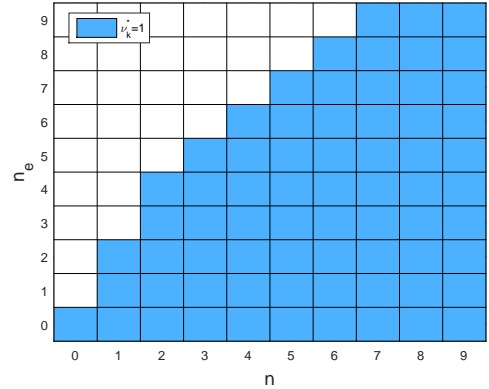} 
\caption{$\nu_k^*$ for different values of $P_{k-1|k-1}=f^n(\bar{P})$ and $P_{e,k-1|k-1}=f^{n_e}(\bar{P})$, at time $k=6$.}
\label{threshold_plot2}
\end{figure}  
 
 Next, we consider the performance as $\beta$ is varied, both when the eavesdropper error covariance is known and  unknown. Fig. \ref{Pk_Pek_plot} plots the trace of the expected error covariance at the estimator $\textrm{tr} \mathbb{E}[P_{k|k}]$ vs. the trace of the expected error covariance at the eavesdropper $\textrm{tr} \mathbb{E}[P_{e,k|k}]$. 
  \begin{figure}[t!]
\centering 
\includegraphics[scale=0.6]{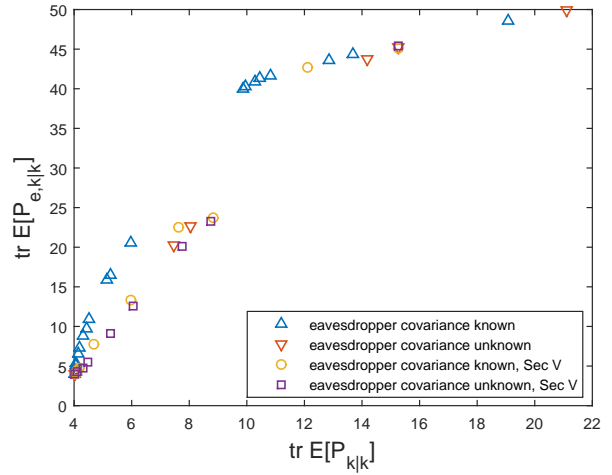} 
\caption{Expected error covariance at estimator vs expected error covariance at eavesdropper. Finite horizon.}
\label{Pk_Pek_plot}
\end{figure} 
Each point is obtained by averaging over 100000 Monte Carlo runs. We see that by varying $\beta$ we obtain a tradeoff between  $\textrm{tr} \mathbb{E}[P_{k|k}]$ and $\textrm{tr} \mathbb{E}[P_{e,k|k}]$, with the tradeoff being better when the eavesdropper error covariance is known. 

 Fig. \ref{Pk_Iek_plot} plots the trace of the expected error covariance $\textrm{tr} \mathbb{E}[P_{k|k}]$ vs. the expected  information $\mathbb{E}[I_{e,k}]$ revealed to the eavesdropper, where $I_{e,k}$ is given by (\ref{directed_mutual_info}).  
 \begin{figure}[t!]
\centering 
\includegraphics[scale=0.6]{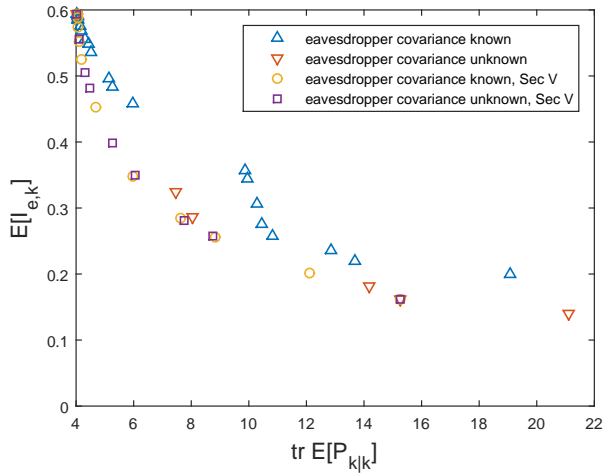} 
\caption{Expected error covariance at estimator vs expected  information revealed to eavesdropper. Finite horizon.}
\label{Pk_Iek_plot}
\end{figure} 
For comparison, the performance obtained by solving problems (\ref{finite_horizon_problem_full_CSI}) and (\ref{finite_horizon_problem_partial_CSI}) in Section \ref{alt_measures_sec} is also plotted in Fig. \ref{Pk_Pek_plot}.   We  observe that the solutions to problems (\ref{finite_horizon_problem_full_CSI}) and (\ref{finite_horizon_problem_partial_CSI}) give worse performance in terms of the tradeoff between $\textrm{tr} \mathbb{E}[P_{k|k}]$ and  $\textrm{tr} \mathbb{E}[P_{e,k|k}]$, but better performance in terms of the tradeoff between $\textrm{tr} \mathbb{E}[P_{k|k}]$ and  $ \mathbb{E}[I_{e,k}]$, since they directly optimize this tradeoff.

\subsection{Infinite Horizon}
\label{numerical_infinite_horizon_sec}
We next present results for the infinite horizon situation.   Table \ref{Pk_Pek_table_infinite_horizon}  tabulates some values of $\textrm{tr} \mathbb{E}[P_{k|k}]$ and $\textrm{tr} \mathbb{E}[P_{e,k|k}]$, obtained by taking the time average of a Monte Carlo run of length 1000000, using the threshold policy in the proof of Theorem \ref{infinite_horizon_unbounded_thm}  which transmits at time $k$ if and only if $P_{k-1|k-1} \geq f^t(\bar{P})$. In the case  $\lambda= 0.6$, $\lambda_e = 0.6$, condition (\ref{lambda_e_condition_infinite_horizon}) for unboundedness of the expected eavesdropper covariance is satisfied when $t \geq 2$, and in the case  $\lambda= 0.6$, $\lambda_e = 0.8$ (where the eavesdropping probability is higher than the packet reception probability), condition (\ref{lambda_e_condition_infinite_horizon}) is satisfied for $t \geq 3$.  We see that in both cases, by using a sufficiently large $t$, one can make the expected error covariance of the eavesdropper very large, while keeping the expected error covariance at the estimator bounded. 
\begin{table}[t!]
\caption{Expected error covariance at estimator vs expected error covariance at eavesdropper. Infinite horizon.}
\centering
\begin{tabular}{|c||c|c||c|c|} \hline
& \multicolumn{2}{|c||}{$\lambda = 0.6$, $\lambda_e = 0.6$} &  \multicolumn{2}{|c|}{$\lambda = 0.6$, $\lambda_e = 0.8$}  \\ \hline
$ t$ & $\textrm{tr} \mathbb{E}[P_{k|k}]$ & $\textrm{tr} \mathbb{E}[P_{e,k|k}]$  & $\textrm{tr} \mathbb{E}[P_{k|k}]$ & $\textrm{tr} \mathbb{E}[P_{e,k|k}]$   \\ \hline \hline
 1 & 5.59 & $ 19.49$ & 5.32 & $ 4.66$ \\ \hline
 2 & 7.53 &  $ 523.06$ & 7.60 &  $ 14.05$ \\ \hline
 3 & 10.76 &  $ 2.82 \times 10^{5}$ & 10.67 &  $ 136.06$ \\ \hline
 4 & 15.36 &  $ 1.21 \times 10^{8}$ & 15.59 &  $ 2.14 \times 10^{3}$ \\ \hline
 5 & 23.57 &  $ 1.19 \times 10^{10}$ & 23.34 &  $ 1.72  \times 10^{5}$ \\ \hline 
 6 & 35.07 &  $ 3.68 \times 10^{13}$ & 35.04 &  $ 6.83  \times 10^{6} $ \\ \hline 
\end{tabular}
\label{Pk_Pek_table_infinite_horizon}
\end{table}

Finally, we consider the performance  obtained by solving the infinite horizon problems (\ref{infinite_horizon_problem_full_CSI}) and (\ref{infinite_horizon_problem_partial_CSI}) as $\beta$ is varied, both when the eavesdropper error covariance is known and  unknown. We use $\lambda= 0.6$, $\lambda_e = 0.6$.  In numerical solutions we use the truncated set $\mathcal{S}^N$ from Remark \ref{SN_remark} with $N=10$. Fig. \ref{Pk_Iek_plot_infinite_horizon} plots the trace of the expected error covariance $\textrm{tr} \mathbb{E}[P_{k|k}]$ vs. the expected  information $\mathbb{E}[I_{e,k}]$ revealed to the eavesdropper, with $\textrm{tr} \mathbb{E}[P_{k|k}]$ and $\mathbb{E}[I_{e,k}]$ obtained by taking the time average of a Monte Carlo run of length 1000000. We again obtain a tradeoff between  $\textrm{tr} \mathbb{E}[P_{k|k}]$ and $\mathbb{E}[I_{e,k}]$. 
However, here  the expected information revealed to the eavesdropper always appears to be bounded away from zero, which is in agreement with Theorem \ref{infinite_horizon_info_thm}. 
\begin{figure}[t!]
\centering 
\includegraphics[scale=0.6]{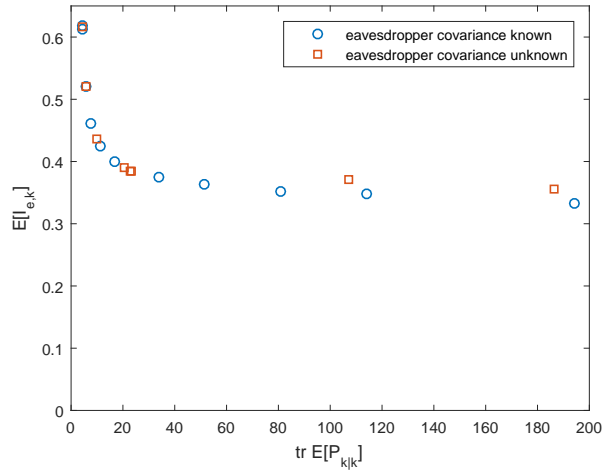} 
\caption{Expected error covariance at estimator vs expected  information revealed to eavesdropper. Infinite horizon.}
\label{Pk_Iek_plot_infinite_horizon}
\end{figure}

\section{Conclusion}
\label{conclusion_sec}
In this paper we have studied the scheduling of sensor transmissions for remote state estimation, where each transmission can be overheard by an eavesdropper with a certain probability. The scheduling is done by solving an optimization problem that minimizes a combination of the expected error covariance at the remote estimator and the negative of the expected error covariance at the eavesdropper. We have derived structural results on the optimal transmission scheduling which show a thresholding behaviour in the optimal policies.  In the infinite horizon situation, we have shown that with unstable systems one can keep the expected estimation error covariance bounded while the expected eavesdropper error covariance becomes unbounded. An alternative measure of security has been considered, where in the infinite horizon situation with unstable systems, we have shown that the expected amount of information revealed to the eavesdropper is always lower bounded away from zero, for any transmission policy which keeps the expected estimation error covariance bounded. Extensions to the basic framework have been considered, such as the transmission of measurements, with an extension to Markovian packet drops currently under investigation. 
Future work will include the investigation of other techniques inspired by  those introduced into physical layer security \cite{ZhouSongZhang}.

\begin{appendix}


\subsection{Proof of Theorem \ref{structural_thm_full_CSI}}
\label{structural_thm_full_CSI_proof}
Define $J_k(\cdot,\cdot): \mathcal{S} \times \mathcal{S} \rightarrow \mathbb{R}$ by:
\begin{align*}
&J_{K+1}(P,P_e)  =0 \\
&J_k(P,P_e)  = \min_{\nu \in \{0,1\}} \Big\{ \beta(\nu \lambda \textrm{tr} \bar{P} + (1-\nu \lambda)\textrm{tr}f(P)) \\ & \quad+ (1-\beta) \nu \lambda_e \Big(\frac{1}{2} \log \det f(P_{e}) - \frac{1}{2} \log \det \bar{P}\Big) \\ & \quad +  \nu \lambda \lambda_e J_{k+1}(\bar{P}, \bar{P}) +  \nu \lambda (1-\lambda_e) J_{k+1}(\bar{P},f(P_e)) \\ & \quad +  \nu (1-\lambda) \lambda_e J_{k+1}(f(P),\bar{P}) 
\\ &  \quad + \big(\nu(1-\lambda)(1-\lambda_e) + 1 - \nu\big)  J_{k+1}(f(P),f(P_e))  \Big\} 
\end{align*}
for $k=K,\dots,1$, 
and define:
\begin{equation}
\label{phi_fn_defn_info}
\begin{split}
& \phi_k(P,P_e)  \triangleq \beta \lambda  \textrm{tr}f(P) - \beta \lambda \textrm{tr} \bar{P} \\ & \quad - (1-\beta) \lambda_e \Big(\frac{1}{2} \log \det f(P_{e}) - \frac{1}{2} \log \det \bar{P}\Big) \\
& \quad + [1-(1-\lambda)(1-\lambda_e)] J_{k+1}(f(P),f(P_e)) \\
& \quad - \lambda \lambda_e J_{k+1}(\bar{P}, \bar{P}) -   \lambda (1-\lambda_e) J_{k+1}(\bar{P},f(P_e)) \\ & \quad -   (1-\lambda) \lambda_e J_{k+1}(f(P),\bar{P}). 
\end{split}
\end{equation}
Part (i) can then  be proved using similar techniques as in the proof of Theorem \ref{structural_thm_full_CSI_alt}(i). 
\\The proof of (ii) requires a more delicate argument involving the use of Lemma \ref{logdet_lemma}. We want to show that for fixed $P$, the function $\phi_k(P,P_e) $ defined by (\ref{phi_fn_defn_info}) is a decreasing function of $P_e$. This will be true if we can show that 
\begin{equation*}
\begin{split}
& -(1-\beta) \lambda_e \frac{1}{2} \log \det f(P_e) + [1-(1-\lambda)(1-\lambda_e)] \\ & \quad \times J_{k}(f(P),f(P_e))-  \lambda (1-\lambda_e) J_{k}(\bar{P},f(P_e)) 
\end{split}
\end{equation*}
is a decreasing function of $P_e$ for all $k$. Using induction, we will prove the slightly more general statement that 
\begin{equation*}
\begin{split}
& -(1-\beta) \lambda_e \frac{1}{2} \log \det f^n(P_e) + [1-(1-\lambda)(1-\lambda_e)] \\ & \quad \times J_{k}(P,f^n(P_e))-  \lambda (1-\lambda_e) J_{k}(P',f^n(P_e)) 
\end{split}
\end{equation*}
is a decreasing function of $P_e$ for all $k$, all $n\in \mathbb{N}$ and all $P, P' \in \mathcal{S}$. 

The case of $k=K+1$ is clear. Now assume that for  $P_e \geq P_e'$, 
\begin{equation}
\label{induction_hypothesis_fixed_P}
\begin{split}
&-(1-\beta) \lambda_e \frac{1}{2} \log \det f^n(P_e) + [1-(1-\lambda)(1-\lambda_e)] \\ & \quad \times  J_{l}(P,f^n(P_e))-  \lambda (1-\lambda_e) J_{l}(P',f^n(P_e)) \\
&  +(1-\beta) \lambda_e \frac{1}{2} \log \det f^n(P_e') - [1-(1-\lambda)(1-\lambda_e)] \\ & \quad \times J_{l}(P,f^n(P_e')) +  \lambda (1-\lambda_e) J_{l}(P',f^n(P_e')) \leq 0
\end{split}
\end{equation}
or equivalently, 
\begin{equation}
\label{induction_hypothesis_equiv_fixed_P}
\begin{split}
&-(1-\beta) \lambda_e \frac{1}{2} \log \det f^n(P_e') + [1-(1-\lambda)(1-\lambda_e)] \\ & \quad \times J_{l}(P,f^n(P_e'))-  \lambda (1-\lambda_e) J_{l}(P',f^n(P_e')) \\
&  +(1-\beta) \lambda_e \frac{1}{2} \log \det f^n(P_e) - [1-(1-\lambda)(1-\lambda_e)]  \\ & \quad \times J_{l}(P,f^n(P_e)) +  \lambda (1-\lambda_e) J_{l}(P',f^n(P_e)) \geq 0
\end{split}
\end{equation}
holds for $l=K+1,K,\dots,k+1$. Then 
\begin{align*}
&-(1-\beta) \lambda_e \frac{1}{2} \log \det f^n(P_e') + [1-(1-\lambda)(1-\lambda_e)] \\ & \quad \times J_{k}(P,f^n(P_e'))-  \lambda (1-\lambda_e) J_{k}(P',f^n(P_e'))  \\
& +(1-\beta) \lambda_e \frac{1}{2} \log \det f^n(P_e) - [1-(1-\lambda)(1-\lambda_e)] \\ & \quad \times J_{k}(P,f^n(P_e)) +  \lambda (1-\lambda_e) J_{k}(P',f^n(P_e)) \\
& \geq -(1-\beta) \lambda_e \frac{1}{2} \log \det f^n(P_e')  +(1-\beta) \lambda_e \frac{1}{2} \log \det f^n(P_e) \\&  +  \min_{\nu \in \{0,1\}} \Bigg\{  [1-(1-\lambda)(1-\lambda_e)] \Big\{ (1-\beta) \nu \lambda_e \\ & \quad \times \frac{1}{2} \log \det f^{n+1}(P_e')   + \nu \lambda (1-\lambda_e) J_{k+1}(\bar{P}, f^{n+1}(P_e')) \\ 
&   \quad + [\nu(1-\lambda)(1-\lambda_e) + 1-\nu] J_{k+1}(f(P),f^{n+1}(P_e'))   \Big\} \\
&  -\lambda (1-\lambda_e)    \Big\{ (1-\beta) \nu \lambda_e \frac{1}{2} \log \det f^{n+1}(P_e') \\ & \quad  + \nu \lambda (1-\lambda_e) J_{k+1}(\bar{P}, f^{n+1}(P_e')) \\
& \quad  + [\nu(1-\lambda)(1-\lambda_e) + 1-\nu] J_{k+1}(f(P'),f^{n+1}(P_e'))   \Big\}  \\ 
&  -[1-(1-\lambda)(1-\lambda_e)] \Big\{ (1-\beta) \nu \lambda_e \frac{1}{2} \log \det f^{n+1}(P_e)  \\ & \quad + \nu \lambda (1-\lambda_e) J_{k+1}(\bar{P}, f^{n+1}(P_e)) \\ 
&  \quad + [\nu(1-\lambda)(1-\lambda_e) + 1-\nu] J_{k+1}(f(P),f^{n+1}(P_e))   \Big\} \\
&   +\lambda (1-\lambda_e)    \Big\{ (1-\beta) \nu \lambda_e \frac{1}{2} \log \det f^{n+1}(P_e)  \\ &  \quad + \nu \lambda (1-\lambda_e) J_{k+1}(\bar{P}, f^{n+1}(P_e)) \\ 
&  \quad + [\nu(1-\lambda)(1-\lambda_e) + 1-\nu] J_{k+1}(f(P'),f^{n+1}(P_e))   \Big\} \Bigg\} \\
& =  \min_{\nu \in \{0,1\}} \Bigg\{  \nu \lambda(1-\lambda_e) \Big \{ [1-(1-\lambda)(1-\lambda_e)] \\ &  \quad \times \!J_{k+1}(\bar{P}, f^{n+1}(P_e')) \!-\! \lambda (1\!-\!\lambda_e)  J_{k+1}(\bar{P}, f^{n+1}(P_e')) \\ 
& \quad  - [1-(1-\lambda)(1-\lambda_e)] J_{k+1}(\bar{P}, f^{n+1}(P_e)) \\ & \quad + \lambda (1-\lambda_e) J_{k+1}(\bar{P}, f^{n+1}(P_e)) \\
& \quad   -(1-\beta) \lambda_e \frac{1}{2} \log \det f^{n+1}(P_e') \\ &  \quad +(1-\beta) \lambda_e \frac{1}{2} \log \det f^{n+1}(P_e) \Big\} \\
&  + [\nu(1-\lambda)(1-\lambda_e) +1-\nu] \Big \{ [1-(1-\lambda)(1-\lambda_e)] \\ & \quad   \times\! J_{k+1}(f(P), f^{n+1}(P_e'))\!-\! \lambda (1\!-\!\lambda_e)  J_{k+1}(f(P'), f^{n+1}(P_e')) \\ 
&  \quad - [1-(1-\lambda)(1-\lambda_e)] J_{k+1}(f(P), f^{n+1}(P_e)) \\ &  \quad + \lambda (1-\lambda_e) J_{k+1}(f(P'), f^{n+1}(P_e)) \\
&  \quad  -(1-\beta) \lambda_e \frac{1}{2} \log \det f^{n+1}(P_e')  \\ &  \quad +(1-\beta) \lambda_e \frac{1}{2} \log \det f^{n+1}(P_e) \Big\} \Bigg\} \\
&   +\!(1\!-\!\beta) \lambda_e \frac{1}{2} \log \det f^{n+1}(P_e')  \!-\! (1\!-\!\beta) \lambda_e \frac{1}{2} \log \det f^{n+1}(P_e) \\
&   -(1-\beta) \lambda_e \frac{1}{2} \log \det f^n(P_e')  +(1-\beta) \lambda_e \frac{1}{2} \log \det f^n(P_e)  \\
&\geq 0
\end{align*}
where the first inequality holds after some cancellation, and the equality is a rearrangement of the terms. The last inequality holds (for both cases $\nu^*=0$ and $\nu^*=1$)
by the induction hypothesis (\ref{induction_hypothesis_equiv_fixed_P}) and the fact that $(1-\beta) \lambda_e \frac{1}{2} \log \det f^n(P_e) - (1-\beta) \lambda_e \frac{1}{2} \log \det f^{n+1}(P_e)$ is an increasing function of $P_e$ by Lemma \ref{logdet_lemma}.

\subsection{Proof of Lemma \ref{logdet_lemma}}
\label{logdet_lemma_proof}
It suffices to show that
$$\log \det P - \log \det f(P) $$ 
is an increasing function of $P$, as the property that $P \geq P' \Rightarrow f^n(P) \geq f^n(P')$ then implies that $\log \det f^n(P) - \log \det f^{n+1}(P) = \log \det f^n(P) - \log \det f(f^{n}(P))$ is an increasing function of $P$. 
Let $\mathbf{S}$ denote the set of all symmetric positive semi-definite matrices. We will use the characterization from pp.108-109 of \cite{BoydVandenberghe} that  a function $F: \mathbf{S} \rightarrow \mathbb{R}$ is matrix monotone increasing (which corresponds to Definition \ref{increasing_fn_defn} when restricted to $\mathcal{S}$) if and only if the gradient $\nabla F(P) \geq 0, \, \forall P \in \mathbf{S}$. We have (see e.g. p.641 of \cite{BoydVandenberghe}) that $\nabla \log \det P = P^{-1}$ (from our assumption that $Q > 0$, we have  $Q$, and hence $P$, being invertible), and by the chain rule that $\nabla \log \det f(P)=\nabla \log \det (A P A^T + Q) = A^T (A P A^T + Q)^{-1} A$. Then
\begin{equation*}
\begin{split}
& \nabla \big(\log \det P - \log \det f(P) \big)   \\
&= P^{-1} - A^T (A P A^T + Q)^{-1} A \\
&=   P^{-1} - P^{-1} P A^T (A P P^{-1} P A^T + Q)^{-1} A P P^{-1}\\
&=    \big(P + P A^T Q^{-1} A P \big)^{-1} \\
& \geq 0 \quad \forall P \in \mathcal{S},
\end{split}
\end{equation*}
where the third equality uses the matrix inversion lemma. 

\subsection{Proof of Lemma \ref{infinite_horizon_upper_bound_lemma}}
\label{infinite_horizon_upper_bound_lemma_proof}
Given a square matrix $X$, let $\sigma_{\min}(X)$ and $\sigma_{\max}(X)$ be the minimum and maximum eigenvalues of $X$ respectively if they are real valued, and let $|\sigma_{\max}(X)|$ be the  spectral radius of $X$. Denote the largest singular value of $X$ by $s_{\max}(X)$.  

We have
\begin{equation*}
\begin{split}
&\frac{1}{N} \left(\frac{1}{2} \log \det f^N(\bar{P}) -  \frac{1}{2} \log \det \bar{P} \right) \\& = \frac{1}{N}\bigg[ \frac{1}{2} \log \det \bigg(A^N \bar{P} (A^N)^T + \sum_{m=0}^{N-1} A^m Q (A^m)^T \bigg) \\ & \quad\quad\quad -  \frac{1}{2} \log \det \bar{P} \bigg] \\
& < \frac{1}{N}  \bigg[ \frac{1}{2} \log \det \bigg(A^N \sigma_{\max}(\bar{P}) I (A^N)^T \\ & \quad \quad \quad + \sum_{m=0}^{N-1} A^m  \sigma_{\max}(Q) I (A^m)^T  \bigg) \bigg] \\
& = \frac{1}{N}  \bigg[ \frac{1}{2} \log \det \bigg(\sigma_{\max}(\bar{P}) A^N   (A^N)^T \\ & \quad \quad \quad + \sum_{m=0}^{N-1} \sigma_{\max}(Q) A^m   (A^m)^T \bigg)  \bigg] 
\end{split}
\end{equation*}
We also have
\begin{equation*}
\begin{split}
\sigma_{\max}  (A^m   (A^m)^T) & \!=\! s_{\max}^2 ((A^m)^T)\! =\! s_{\max}^2 ((A^m)) \!\leq\!  s_{\max}^{2m} (A)
\end{split}
\end{equation*}
where the inequality follows from e.g. p.454 of \cite{HornJohnson}. Then by  Weyl's Theorem \cite[p.239]{HornJohnson}, all eigenvalues of $\sigma_{\max}(\bar{P}) A^N   (A^N)^T + \sum_{m=0}^{N-1} \sigma_{\max}(Q) A^m   (A^m)^T$ will be less than 
\begin{equation*}
\begin{split}
&\sigma_{\max}(\bar{P}) s_{\max}^{2N} (A) + \sum_{m=0}^{N-1} \sigma_{\max}(Q) s_{\max}^{2m} (A) \\& = \sigma_{\max}(\bar{P}) s_{\max}^{2N} (A) + \frac{ \sigma_{\max}(Q) (s_{\max}^{2N} (A)-1) }{s_{\max}^{2} (A)-1} \\
& < \left(\sigma_{\max}(\bar{P}) + \frac{ \sigma_{\max}(Q) }{s_{\max}^{2} (A)-1} \right) s_{\max}^{2N} (A)
\end{split}
\end{equation*}
Recalling that $n_x$ is the dimension of $x_k$, we thus have
\begin{equation*}
\begin{split}
& \frac{1}{N} \left(\frac{1}{2} \log \det f^N(\bar{P}) -  \frac{1}{2} \log \det \bar{P} \right) \\ & < \frac{1}{N}  \bigg[ \frac{1}{2} \log \det \bigg(\sigma_{\max}(\bar{P}) A^N   (A^N)^T \\ & \quad \quad \quad + \sum_{m=0}^{N-1} \sigma_{\max}(Q) A^m   (A^m)^T \bigg)  \bigg] \\
& = \frac{1}{N}  \bigg[ \frac{1}{2} \log \prod_i \sigma_i  \Big(\sigma_{\max}(\bar{P}) A^N   (A^N)^T \\ & \quad \quad \quad + \sum_{m=0}^{N-1} \sigma_{\max}(Q) A^m   (A^m)^T \Big)  \bigg] \\
& < \frac{1}{N}  \left[ \frac{1}{2} \log \left(\left(\sigma_{\max}(\bar{P}) + \frac{ \sigma_{\max}(Q) }{s_{\max}^{2} (A)-1} \right)^{n_x} s_{\max}^{2Nn_x} (A) \right) \right] \\
& \leq  \frac{1}{2} \log \left(\sigma_{\max}(\bar{P}) + \frac{ \sigma_{\max}(Q) }{s_{\max}^{2} (A)-1} \right)^{n_x} + n_x \log s_{\max}(A) \\ & \triangleq \Delta_U < \infty.
\end{split}
\end{equation*}

\subsection{Proof of Lemma \ref{infinite_horizon_lower_bound_lemma}}
\label{infinite_horizon_lower_bound_lemma_proof}
The proof is similar to, though slightly more involved than, the proof of Lemma \ref{infinite_horizon_upper_bound_lemma}. 
We have
\begin{equation*}
\begin{split}
&\frac{1}{N} \left(\frac{1}{2} \log \det f^N(\bar{P}) -  \frac{1}{2} \log \det \bar{P} \right) \\ & = \frac{1}{N}\bigg[ \frac{1}{2} \log \det \bigg(A^N \bar{P} (A^N)^T + \sum_{m=0}^{N-1} A^m Q (A^m)^T \bigg) \\ & \quad \quad \quad -  \frac{1}{2} \log \det \bar{P} \bigg] \\
& \geq \frac{1}{N}  \left[ \frac{1}{2} \log \det \left(A^N \bar{P} (A^N)^T +Q \right) -  \frac{1}{2} \log \det \bar{P} \right] \\
& \geq \frac{1}{N}  \left[ \frac{1}{2} \log \det \left(\sigma_{\min}(\bar{P}) A^N   (A^N)^T +Q \right) -  \frac{1}{2} \log \det \bar{P} \right] 
\end{split}
\end{equation*}

We also have
\begin{equation*}
\begin{split}
& \sigma_{\max}  (A^N   (A^N)^T) = s_{\max}^2 ((A^N)^T) \\ & \quad \geq |\sigma_{\max} ((A^N)^T)|^2   = |\sigma_{\max} (A)|^{2N}
\end{split}
\end{equation*}
where the inequality now follows from e.g. p.347 of \cite{HornJohnson}. By  Weyl's Theorem, the largest eigenvalue of $\sigma_{\min}(\bar{P}) A^N   (A^N)^T +Q$ will be greater than $\sigma_{\min}(\bar{P}) |\lambda_{\max} (A)|^{2N}$, and the remaining eigenvalues of $\sigma_{\min}(\bar{P}) A^N   (A^N)^T +Q$ will be greater than $\sigma_{\min} (Q)$. 
Thus
\begin{equation*}
\begin{split}
& \frac{1}{N} \left(\frac{1}{2} \log \det f^N(\bar{P}) -  \frac{1}{2} \log \det \bar{P} \right) \\  & \geq  \frac{1}{N}  \left[ \frac{1}{2} \log  \prod_i \sigma_i \Big(\sigma_{\min}(\bar{P}) A^N   (A^N)^T +Q \Big) -  \frac{1}{2} \log \det \bar{P} \right] \\
& \geq \frac{1}{N}  \left[ \frac{1}{2} \log   \left(\sigma_{\min}(\bar{P}) |\sigma_{\max} (A)|^{2N} \sigma_{\min}^{n_x-1} (Q)  \right) -  \frac{1}{2} \log \det \bar{P} \right] \\
& =   \log( |\sigma_{\max} (A)|) + \frac{1}{N}  \bigg[ \frac{1}{2} \log   \left(\sigma_{\min}(\bar{P})  \sigma_{\min}^{n_x-1} (Q)  \right) \\ & \quad \quad -  \frac{1}{2} \log \det \bar{P} \bigg]
\end{split}
\end{equation*}
 Let $N'$ be sufficiently large that 
\begin{equation*}
\begin{split}
 \Delta_1 \triangleq & \log( |\sigma_{\max} (A)|) + \frac{1}{N'}  \bigg[ \frac{1}{2} \log   \left(\sigma_{\min}(\bar{P})  \sigma_{\min}^{n_x-1} (Q)  \right) \\ & \quad -  \frac{1}{2} \log \det \bar{P} \bigg]  > 0, 
 \end{split}
 \end{equation*}
which can be satisfied since $ |\sigma_{\max} (A)| > 1$. 
Then we have $\frac{1}{N} \left(\frac{1}{2} \log \det f^N(\bar{P}) -  \frac{1}{2} \log \det \bar{P} \right) > \Delta_1$ for all $N \geq N'$. 

Next, since  $f^N(\bar{P}) \geq \bar{P} > 0, \forall N \in \mathbb{N}$, and $f^N(\bar{P}) \neq \bar{P}$, one can use Theorem 8.4.9 of \cite{Bernstein_book} to conclude that   $\det f^N(\bar{P}) > \det \bar{P}, \forall N \in \mathbb{N}$.
Letting 
$$\Delta_2 \triangleq \!\! \min_{N \in \{1,2,\dots,N'\}} \frac{1}{N} \! \left(\frac{1}{2} \log \det f^N(\bar{P}) -  \frac{1}{2} \log \det \bar{P} \right) > 0,$$
we have  $\frac{1}{N} \left(\frac{1}{2} \log \det f^N(\bar{P}) -  \frac{1}{2} \log \det \bar{P} \right) > \Delta_2$ for all $N \leq N'$. Defining $\Delta_L = \min(\Delta_1,\Delta_2)$ then gives the result. 

\subsection{Proof of Theorem \ref{structural_thm_full_CSI_tx_meas}}
\label{structural_thm_full_CSI_tx_meas_proof}
In order to prove Theorem \ref{structural_thm_full_CSI_tx_meas},  the following result from \cite{LeongDeyQuevedo_TAC} will also be required: 
\begin{lemma}[From \cite{LeongDeyQuevedo_TAC}]
\label{composition_lemma}
Suppose the system is scalar. Let $\mathcal{F}(.)$ be a function formed by composition (in any order) of any of the functions $f(.), g(.), \textnormal{id}(.)$
where 
$$f(P) = A^2 P + Q, \quad g(P) = A^2 P + Q - \frac{A^2 C^2 P^2}{C^2 P + R},$$ and $\textnormal{id}(.)$ is the identity function. Then  $ \mathcal{F}(f(P)) - \mathcal{F}(g(P))$ is an increasing function of $P$. 
\end{lemma}

(i) Define the functions  $J_k(\cdot,\cdot): \mathbb{R} \times \mathbb{R} \rightarrow \mathbb{R}$ by:
\begin{equation}
\label{J_fn_defn_tx_meas}
\begin{split}
&J_{K+1}(P,P_e)  =0 \\
&J_k(P,P_e)  = \min_{\nu \in \{0,1\}} \Big\{ \beta(\nu \lambda \textrm{tr}g(P) + (1-\nu \lambda)\textrm{tr}f(P)) \\ & \quad -  (1-\beta) (\nu \lambda_e \textrm{tr} g(P_e) + (1-\nu \lambda_e)\textrm{tr}f(P_e)) \\ & \quad +  \nu \lambda \lambda_e J_{k+1}(g(P),g(P_e)) +  \nu \lambda (1-\lambda_e) J_{k+1}(g(P),f(P_e))  \\ & \quad +  \nu (1-\lambda) \lambda_e J_{k+1}(f(P),g(P_e)) 
\\ &  \quad + \big(\nu(1-\lambda)(1-\lambda_e) + 1 - \nu\big)  J_{k+1}(f(P),f(P_e))  \Big\} 
\end{split}
\end{equation}
for $k=K,\dots,1$. Then problem  (\ref{finite_horizon_problem_full_CSI_tx_meas}) is solved by computing $J_k(P_{k|k-1}, P_{e,k|k-1})$ for $k = K,K-1,\dots,1$.  For scalar systems, define:
\begin{align}
\label{phi_fn_defn_tx_meas}
& \phi_k(P,P_e)  \nonumber \\ & \triangleq \beta \lambda f(P) - \beta \lambda g(P)  - (1-\beta) \lambda_e f(P_e) + (1-\beta) \lambda_e g(P_e) \nonumber \\
& \quad + [1-(1-\lambda)(1-\lambda_e)] J_{k+1}(f(P),f(P_e)) \nonumber\\ & \quad - \lambda \lambda_e J_{k+1}(g(P), g(P_e)) -   \lambda (1-\lambda_e) J_{k+1}(g(P),f(P_e)) \nonumber \\ & \quad  -   (1-\lambda) \lambda_e J_{k+1}(f(P),g(P_e)) 
\end{align}
As in the proof of Theorem \ref{structural_thm_full_CSI_alt}(i), we wish to show that for fixed $P_e$, $\phi_k(P,P_e)$ is an increasing function of $P$. Since $f(P)$ and $g(P)$ are increasing functions of $P$, this will be the case if we can show that 
\begin{equation*}
\begin{split}
& [1\!-\!(1\!-\!\lambda)(1\!-\!\lambda_e)] J_{k}(\mathcal{F}f(P),P_e) \!-\! \lambda \lambda_e J_{k}(\mathcal{F}g(P),P_e') \\ &  -\! \lambda (1\!-\!\lambda_e) J_{k}(\mathcal{F}g(P),P_e) \!-\!  (1\!-\!\lambda) \lambda_e J_{k}(\mathcal{F}f(P),P_e') 
\end{split}
\end{equation*}
is an increasing function of $P$ for all $k$, all $\mathcal{F}(.)$ formed by compositions of the functions $f(.), g(.), \textrm{id}(.)$, and all $P_e, P_e' \in \mathcal{S}$. The case of $k=K+1$ is clear. Now assume that for $P \geq P'$, 
\begin{equation}
\label{induction_hypothesis_tx_meas}
\begin{split}
& [1\!-\!(1\!-\!\lambda)(1\!-\!\lambda_e)] J_{l}(\mathcal{F}f(P),P_e) \!-\! \lambda \lambda_e J_{l}(\mathcal{F}g(P),P_e') \\ & \quad -\! \lambda (1\!-\!\lambda_e) J_{l}(\mathcal{F}g(P),P_e) \!-\!  (1\!-\!\lambda) \lambda_e J_{l}(\mathcal{F}f(P),P_e') \\ 
& - \![1\!-\!(1\!-\!\lambda)(1\!-\!\lambda_e)] J_{l}(\mathcal{F}f(P'),P_e) \!+\! \lambda \lambda_e J_{l}(\mathcal{F}g(P'),P_e') \\ & \quad +\! \lambda (1\!-\!\lambda_e) J_{l}(\mathcal{F}g(P'),P_e) \!+\!  (1\!-\!\lambda) \lambda_e J_{l}(\mathcal{F}f(P'),P_e') \\ & \geq 0
\end{split}
\end{equation}
holds for $l=K+1,K,\dots,k+1$. Then
\begin{align*}
& [1\!-\!(1\!-\!\lambda)(1\!-\!\lambda_e)] J_{k}(\mathcal{F}f(P),P_e) \!-\! \lambda \lambda_e J_{k}(\mathcal{F}g(P),P_e') \\ & \quad -\! \lambda (1\!-\!\lambda_e) J_{k}(\mathcal{F}g(P),P_e) \!-\!  (1\!-\!\lambda) \lambda_e J_{k}(\mathcal{F}f(P),P_e') \\ 
& - \![1\!-\!(1\!-\!\lambda)(1\!-\!\lambda_e)] J_{k}(\mathcal{F}f(P'),P_e) \!+\! \lambda \lambda_e J_{k}(\mathcal{F}g(P'),P_e') \\ & \quad +\! \lambda (1\!-\!\lambda_e) J_{k}(\mathcal{F}g(P'),P_e) \!+\!  (1\!-\!\lambda) \lambda_e J_{k}(\mathcal{F}f(P'),P_e') \\
& \geq \min_{\nu \in \{0,1\}} \Bigg\{  [1-(1-\lambda)(1-\lambda_e)] \Big\{ \beta [ \nu \lambda g\mathcal{F}f(P)  \\ & \quad + (1-\nu \lambda) 
f\mathcal{F}f(P)] - (1-\beta)  [ \nu \lambda_e g(P_e) + (1-\nu \lambda_e) f(P_e)] \\
& \quad  + \nu \lambda \lambda_e J_{k+1}(g\mathcal{F}f(P),g(P_e)) \\ & \quad + \nu \lambda(1-\lambda_e) J_{k+1} (g\mathcal{F}f(P),f(P_e)) \\ & \quad 
+\nu(1-\lambda) \lambda_e J_{k+1} (f\mathcal{F}f(P),g(P_e)) \\
& \quad  + [\nu(1-\lambda)(1-\lambda_e) + (1-\nu)]J_{k+1} (f\mathcal{F}f(P),f(P_e)) \Big\} \\
&  \quad \vdots \\ 
& + (1-\lambda) \lambda_e   \Big\{ \beta [ \nu \lambda g\mathcal{F}f(P') + (1-\nu \lambda) f\mathcal{F}f(P')] \\ & \quad 
- (1-\beta)[ \nu \lambda_e  g(P_e') + (1-\nu \lambda_e)  f(P_e')] \\
& \quad + \nu \lambda \lambda_e J_{k+1}(g\mathcal{F}f(P'),g(P_e')) \\ & \quad + \nu \lambda(1-\lambda_e) J_{k+1} (g\mathcal{F}f(P'),f(P_e')) \\ & \quad 
+\nu(1-\lambda) \lambda_e J_{k+1} (f\mathcal{F}f(P'),g(P_e')) \\
& \quad + [\nu(1-\lambda)(1-\lambda_e) + (1-\nu)]J_{k+1} (f\mathcal{F}f(P'),f(P_e')) \Big\}  \Bigg\} 
\end{align*}
Since $f\mathcal{F}(.)$ and $g\mathcal{F}(.)$ are also compositions of functions of the form $f(.),g(.),\textrm{id}(.)$, Lemma \ref{composition_lemma} and the induction hypothesis (\ref{induction_hypothesis_tx_meas}) can be used to conclude, after some calculations, that the above is positive. 
\\(ii) We now wish to show that 
\begin{equation*}
\begin{split}
& [1\!-\!(1\!-\!\lambda)(1\!-\!\lambda_e)] J_{k}(P,\mathcal{F}f(P_e)) \!-\! \lambda \lambda_e J_{k}(P',\mathcal{F}g(P_e)) \\ &  -\! \lambda (1\!-\!\lambda_e) J_{k}(P',\mathcal{F}f(P_e)) \!-\!  (1\!-\!\lambda) \lambda_e J_{k}(P,\mathcal{F}g(P_e)) 
\end{split}
\end{equation*}
is an decreasing function of $P_e$ for all $k$, all $\mathcal{F}(.)$ formed by compositions of the functions $f(.), g(.), \textrm{id}(.)$, and all $P, P' \in \mathcal{S}$.
The proof is similar to part (i).

\end{appendix}

\bibliography{IEEEabrv,secure_remote_estimation}
\bibliographystyle{IEEEtran}

\end{document}